\documentclass[11pt]{amsart}

\usepackage{amsmath}
\usepackage{amsthm}
\usepackage{amsfonts}
\usepackage{graphicx}
\usepackage{psfrag}

\newcommand{\del}{{\partial}}
\newcommand{\rhoconst}{{\rho}}
\newcommand{\epsconst}{{\varepsilon}}

\newtheorem{prop}{Proposition}

\newtheorem*{theorem}{Theorem}

\begin{document}
\title[MTTs generated from nonlinear scalar field data]{Marginally trapped tubes generated from nonlinear scalar field initial data}
\author{Catherine Williams} 
\address{Department of Mathematics, Stanford University, Stanford, CA 94305}
\email{cathwill@math.stanford.edu}      
\date{\today}          

\begin{abstract}We show that the maximal future development of asymptotically flat spherically symmetric black hole initial data for a self-gravitating nonlinear scalar field, also called a Higgs field, contains a connected, achronal, spherically symmetric marginally trapped tube which is asymptotic to the event horizon of the black hole, provided the initial data is sufficiently small and decays like $O(r^{-\frac{1}{2}})$, and the potential function $V$ is nonnegative with bounded second derivative. 
This result can be loosely interpreted as a statement about the stability of `nice' asymptotic behavior of marginally trapped tubes under certain small perturbations of Schwarzschild.  
\end{abstract}

\maketitle

\section{Introduction}

Black holes lie at the core of many current efforts to further our understanding of gravitation.
Questions surrounding their existence and properties are central to  some of the most significant open problems in mathematical relativity,  
while a major focus in numerical relativity is locating evolving black holes in simulations,
and considerable research in physics communities is dedicated to carrying over concepts from quantum mechanics to black hole regimes.
In all of these contexts, certain spacetime hypersurfaces known as marginally trapped tubes (MTTs) play an important role.  
(Marginally trapped tubes of certain causal characters are often referred to as dynamical or isolated horizons in the physics literature.) 
On one hand, mathematically, these hypersurfaces generally lie inside of black holes and can be thought of as forming boundaries between the regions of weak and strong gravitational fields. 
Understanding their behavior thus sheds some light on the portions of black holes' interiors in which singularities and/or Cauchy horizons may form.
On the other hand, numerical relativists and physicists, e.g.\ those developing loop quantum gravity, have largely set aside the traditionally-defined event horizon and instead use MTTs as models of surfaces of black holes \cite{SKB, A}; the advantage of the latter is primarily that they are defined quasi-locally, whereas the former notion requires global information.
In any case, whether interpreting MTTs as black hole surfaces or just interesting structures inside them, it is useful to characterize their long term behavior and its relationship with classical event horizons.

There is a general expectation that MTTs that form during gravitational collapse become spacelike or null and asymptotically approach the event horizon.
Essentially nothing is currently known about the asymptotic behavior of MTTs in general, however, i.e.\ without symmetry.
Indeed, it follows from the existence results of Andersson, Mars, and Simon \cite{AMS, AMS2} that a general black hole spacetime may contain uncountably many distinct MTTs, and in fact, there may be open sets of points having the property that \emph{each point}  lies on  uncountably many distinct MTTs.  
It is an open problem whether any one of these MTTs need be asymptotic to the event horizon, much less all of them.
Imposing spherical symmetry, however, the problem becomes simpler, since each point can lie in at most one spherically symmetric MTT, i.e.\ one foliated by round two-spheres.
We can then ask whether this particular MTT has the expected ``nice" asymptotic behavior.
Henceforth we shall focus our attention solely on such spherically symmetric MTTs.

Even in spherical symmetry, a major problem in trying to compare the asymptotic behavior of MTTs with classical event horizons is that the latter are difficult to locate in general, since their definition requires global information. 
One can start by examining known exact black hole solutions to Einstein's equations, but unfortunately this list of spacetimes is rather quickly exhausted. 
The exact spherically symmetric black hole solutions are the Schwarzschild, Reisner-Nordstr\"om, and Vaidya spacetimes, whose matter models are vacuum, electro-vacuum, and ingoing null dust, respectively. 
In all three of these, the (spherically symmetric) MTTs do exhibit the nice asymptotic behavior just described. 
In fact, in Schwarzschild and Reisner-Nordstr\"om, the MTTs coincide exactly with the black hole event horizons.
In Vaidya, provided the dominant energy condition holds, the MTT is achronal and is either asymptotic to or eventually coincides with the event horizon \cite{A, Wil}.

General spherically symmetric black hole spacetimes satisfying only the dominant energy condition were considered in \cite{Wil}, and it was shown there that if four particular inequalities involving the metric and stress-energy tensor are satisfied in a certain small region near future timelike infinity $i^+$, the future limit point of the event horizon, then the black hole contains an MTT which eventually becomes achronal and asymptotically approaches the event horizon.
It was also shown that these inequalities are satisfied in two families of self-gravitating Higgs field black holes.
The Higgs field matter model is generated by a scalar field $\phi$ satisfying the nonlinear (coupled) wave equation $\Box_g \phi = V^\prime(\phi)$, where $V$ is some given potential function; by Higgs field black holes, we mean spacetimes which are assumed \emph{a priori} to contain black holes.
In particular, the hypotheses of the general theorem were shown to hold provided that certain quantities involving the metric, $\phi$, and $V(\phi)$ either: one, decay at specified rates along the event horizon, or two, satisfy several rather restrictive monotonicity properties along the event horizon at late times.

Given a particular matter model, an alternate, more physical strategy is to recast the Einstein equations as an initial value problem and generate spacetimes from initial data, then locate any black holes and look for MTTs inside them. 
This program has been successfully carried out for several matter models.
In particular, the maximal development of spherically symmetric asymptotically flat initial data for the Einstein equations coupled with a scalar field, the Maxwell equations and a real scalar field, or the Vlasov equation (describing a collisionless gas) does indeed contain an
MTT which is asymptotic to the event horizon \cite{C2}, \cite{D1, DRod}, \cite{DR}.
Furthermore, in the scalar field cases, the MTT is necessarily achronal.

In this paper we reconsider the problem for Higgs field spacetimes, this time taking the latter approach and beginning with asymptotically flat spherically symmetric Higgs field initial data. 
In order to circumvent the difficulties associated with the global existence problem for the Einstein equations, we rely on the general results of \cite{D2} to tell us that as long as our data contains a trapped surface, contains no weakly anti-trapped surfaces,  and generates a spacetime with nonempty future null infinity, $\mathcal{I}^+ \neq \emptyset$, then the maximal development of the data will contain a black hole.
We make two main additional assumptions: one, that the scalar field decays to some limiting value like $O(r^{-\frac{1}{2}})$, where $r$ is a radial coordinate which tends to infinity on the asymptotically flat end; and two, that the scalar field and the mass flux are sufficiently small outside of the outermost marginally trapped sphere. 
We then show that the black hole generated by these initial data contains an achronal, spherically symmetric MTT which is asymptotic to the event horizon, and moreover, this MTT is smooth and connected, intersecting the initial hypersurface at precisely the outermost marginally trapped surface.

A few remarks are in order.
First, this result is considerably stronger than those of \cite{Wil}, since the Higgs field black hole generated under the assumptions just described need not satisfy either set of conditions along the event horizon mentioned previously; the latter are much more restrictive.
Secondly, it should be emphasized that, unlike previous results, here we locate the whole (connected) future development of the MTT emanating  from the Cauchy surface, rather than just a small portion of it near $i^+$.
And thirdly, this result can be loosely interpreted as a statement about certain small perturbations of Schwarzschild.  
Suppose we take a spherically symmetric, spacelike slice of a Schwarzschild spacetime which extends from an inner boundary inside the black hole region out to spacelike infinity $i^0$, and suppose we perturb the vacuum metric on this slice to one given by a very small, decaying Higgs field, say with compact support.
Then the result of this paper says that, while the perturbed MTT may not coincide exactly with the event horizon as it does in Schwarzschild, it will still be achronal and connected and will asymptotically approach the event horizon from inside the black hole.

%%%%%%%%%%%%%%%%%%%%%%%%%%%%%%%%%%%%%%%%%%%%%%%%%%%%%%%%%%%%%%%%%%%%%%%%%%%%%%%%%%%%%%%%%%%%%%%%%%%%%%%%%%%%%%%%%%%%
%
%
%						PRELIMINARIES
%
%
%%%%%%%%%%%%%%%%%%%%%%%%%%%%%%%%%%%%%%%%%%%%%%%%%%%%%%%%%%%%%%%%%%%%%%%%%%%%%%%%%%%%%%%%%%%%%%%%%%%%%%%%%%%%%%%%%%%%

\section{Preliminaries}\label{prelim}

The spacetimes we wish to consider are characterized by the following four properties: their matter is described by a Higgs field matter model satisfying the dominant energy condition; they are spherically symmetric; they arise evolutionarily, i.e.\ as maximal future developments of asymptotically flat initial Cauchy data prescribed on spacelike hypersurfaces; and the Cauchy data from which they arise is physically reasonable black hole data.  We describe the details of each of these requirements below.  It should be noted that we do not deal with the Cauchy problem and the issue of specifying initial data \emph{per se}.  Instead we posit those conditions which are necessary to address the situation at hand, i.e.\ the evolution of a spherically symmetric black hole, and then in the statement of the main theorem identify further, more restrictive conditions on the data which guarantee the existence and desired behavior of a unique spherically symmetric marginally trapped tube inside the black hole.

%%%%%%%%%%%%%%%%%%%%%%%%%%%%%%%%%%%%%%%%%%%%%%%%%%%%%%%%%%%%%%%%%%%%%%%%%%%%%%%%%%%%%%%%%%%%%%%%%%%%%%%%%%%%%%%%%
%
%
%					HIGGS MATTER MODEL & DEC
%
%%%%%%%%%%%%%%%%%%%%%%%%%%%%%%%%%%%%%%%%%%%%%%%%%%%%%%%%%%%%%%%%%%%%%%%%%%%%%%%%%%%%%%%%%%%%%%%%%%%%%%%%%%%%%%%%%%

\subsection{Higgs field matter model \& dominant energy condition}\label{mmDEC}
The Higgs field matter model on a spacetime $(\mathcal{M},g)$ consists of a scalar function $\phi \in C^2(\mathcal{M})$ and a potential function $V(\phi)$, $V \in C^2(\mathbb{R})$,  such that
\begin{equation}\label{eeq} R_{\alpha \beta} - \textstyle \frac{1}{2} R g_{\alpha \beta} = 2 T_{\alpha \beta} \end{equation} 
\begin{equation}\label{sfT} T_{\alpha \beta} = \phi_{; \alpha} \phi_{; \beta} - \left( \textstyle\frac{1}{2} \phi_{;\gamma} \phi^{;\gamma} + V(\phi) \right) g_{\alpha \beta} \end{equation}
and
\begin{equation}\label{sfeq} \Box_g \phi = g^{\alpha \beta}\phi_{;\alpha \beta}  = V^\prime(\phi). \end{equation}
The dominant energy condition stipulates that, for a vector field $\xi^\alpha$ on $\mathcal{M}$,
$-T^\alpha_\beta \xi^\beta$ is future causal wherever $\xi^\alpha$  is future causal.
Using the spherical symmetry of the metric described below, namely the decomposition of the metric given by \eqref{metric}, one readily computes that the dominant energy condition is satisfied if and only if $ V(\phi) \geq 0$ everywhere on $\mathcal{M}$.  We therefore assume \emph{a priori} that the potential function $V$ is an everywhere nonnegative function of its argument.

%%%%%%%%%%%%%%%%%%%%%%%%%%%%%%%%%%%%%%%%%%%%%%%%%%%%%%%%%%%%%%%%%%%%%%%%%%%%%%%%%%%%%%%%%%%%%%%%%%%%%%%%%%%%%%%%%
%
%
%					SPHERICAL SYMMETRY
%
%%%%%%%%%%%%%%%%%%%%%%%%%%%%%%%%%%%%%%%%%%%%%%%%%%%%%%%%%%%%%%%%%%%%%%%%%%%%%%%%%%%%%%%%%%%%%%%%%%%%%%%%%%%%%%%%%%

\subsection{Spherical symmetry}
A self-gravitating Higgs field spacetime (or Cauchy surface) is said to be spherically symmetric if the Lie group $SO(3)$ acts on it by isometries under which $\phi$ remains invariant, with orbits which are either spacelike two-spheres or fixed points.  We shall in fact make a slightly stronger assumption, that in addition to $\mathcal{M}$ admitting such an $SO(3)$-action, the quotient $\mathcal{Q} = \mathcal{M}/SO(3)$ inherits from $\mathcal{M}$ the structure of a 1+1-dimensional Lorentzian manifold, possibly with boundary, with metric $\overline{g}$ such that
\begin{equation} g = \overline{g} + r^2 \gamma. \label{metric}\end{equation}
Here $\gamma$ is the usual round metric on $S^2$, and $r$ is a smooth nonnegative function on $\mathcal{Q}$, called the \emph{area-radius},  whose value at each point is proportional to the square root of the area of the corresponding two-sphere upstairs in $\mathcal{M}$.  Since it is preserved by the $SO(3)$-action, the function $\phi$ descends to a function $\mathcal{Q}$.  The advantage of such an assumption on the set $\mathcal{Q}$ is that such features of $(\mathcal{M},g)$ as black holes, event horizons, and spherically symmetric marginally trapped tubes are preserved and may be studied at this quotient level. 

We further assume that $\mathcal{Q}$ admits a conformal embedding into a subset of 2-dimensional Minkowski space $\mathbb{M}^{1+1}$.  Such an embedding preserves causal structure, so identifying $\mathcal{Q}$ with its image under this embedding, we make use of the usual global double-null coordinates $u$, $v$ on $\mathbb{M}^{1+1}$ to write 
\[ \overline{g} = - \Omega^2 du dv, \]
where $\Omega = \Omega(u,v) > 0$ on $\mathcal{Q}$.  
(In our Penrose diagrams, we will always depict the positive $u$- and $v$-axes at $135^\circ$ and $45^\circ$ from the usual positive $x$-axis, respectively.) 
Since now $\phi = \phi(u,v)$ and $r = r(u,v)$ as well, 
%we compute that
%\begin{equation*} T_{uu} = (\del_u\phi)^2, \quad T_{vv} = (\del_v\phi)^2, \quad
%\text{ and } \quad T_{uv} = \textstyle\frac{1}{2}\Omega^2 V(\phi) \end{equation*}
%and 
we can rewrite \eqref{eeq}-\eqref{sfeq} as a system of pointwise equations on $\mathcal{Q}$:
\begin{eqnarray}  
\del_u(\Omega^{-2} \del_u r) & = & - r \Omega^{-2} (\del_u\phi)^2  \label{eq1}\\
\del_v(\Omega^{-2} \del_v r) & = & - r \Omega^{-2} (\del_v\phi)^2 \label{eq2}\\
\del_u m & = & r^2 \left( V(\phi) \del_u r - 2 \Omega^{-2} (\del_u\phi)^2 \del_v r \right) \label{eq3}\\
\del_v m & = & r^2 \left( V(\phi) \del_v r - 2 \Omega^{-2} (\del_v\phi)^2 \del_u r \right) \label{eq4}  
\end{eqnarray}
and
\begin{equation}\label{sfequv} V^\prime(\phi) = -4\Omega^{-2} \left( \del^2_{uv}\phi + \del_u\phi\,(\del_v \log r) + \del_v\phi\,(\del_u \log r) \right), \end{equation}
where
\begin{equation} m = m(u,v) = \frac{r}{2}\left( 1 - \overline{g}\left( \nabla r, \nabla r \right) \right) = \frac{r}{2}(1 + 4\Omega^{-2} \del_u r \del_v r) \label{mdef} \end{equation}
is the \emph{Hawking mass}.
Note that the null constraints (\ref{eq1}) and (\ref{eq2}) are just Raychaudhuri's equation applied to each of the two null directions in $\mathcal{Q}$.  
Since we can pass back and forth between $(\mathcal{M}, g, \phi)$  and $(\mathcal{Q}, \Omega, r, \phi)$ without losing information, we may work directly on $\mathcal{Q}$ without any loss of generality.

%%%%%%%%%%%%%%%%%%%%%%%%%%%%%%%%%%%%%%%%%%%%%%%%%%%%%%%%%%%%%%%%%%%%%%%%%%%%%%%%%%%%%%%%%%%%%%%%%%%%%%%%%%%%%%%%%
%
%
%				EVOLUTION FROM AF INITIAL DATA
%
%%%%%%%%%%%%%%%%%%%%%%%%%%%%%%%%%%%%%%%%%%%%%%%%%%%%%%%%%%%%%%%%%%%%%%%%%%%%%%%%%%%%%%%%%%%%%%%%%%%%%%%%%%%%%%%%%%

\subsection{Evolution from asymptotically flat initial data}\label{AFevol}
Next we assume that our spacetime arises as the future evolution of initial data for the Einstein-Higgs system \eqref{eeq}-\eqref{sfeq}.  
In particular, we require that $(\mathcal{M},g)$ be the maximal future Cauchy development of initial data prescribed on a spherically symmetric spacelike hypersurface $\Sigma \subset \mathcal{M}$.
Since $\Sigma$ is preserved under the $SO(3)$-action, it descends to a spacelike curve $\mathcal{S} \subset \mathcal{Q}$, and our assumption upstairs implies that on the quotient level we have $\mathcal{Q} = D^+(\mathcal{S})$; in particular, $\mathcal{S}$ is the past boundary of $\mathcal{Q}$. 
The remaining assumptions on initial data may be formulated directly on $\mathcal{S}$.

We further assume that the initial hypersurface $\mathcal{S}$ has at least one asymptotically flat end, and for simplicity, we focus our attention on a single such end.   
Here asymptotic flatness means first of all that $\mathcal{S}$ is a connected spacelike curve along which $r \rightarrow \infty$ in one direction, say the direction of increasing $v$; the other end may or may not have boundary.
Without loss of generality we assume that $r$ is strictly positive on $\mathcal{S}$.
Moreover, we require that the metric and stress-energy tensor approach the Euclidean and vacuum ones, respectively, as $r \rightarrow \infty$ along $\mathcal{S}$.
By inspection of \eqref{sfT}, it follows that for Higgs initial data, $\nabla \phi$ and  $V(\phi) \rightarrow 0$ along $\mathcal{S}$.
We therefore assume that $\phi \rightarrow \phi_+$ along $\mathcal{S}$ and that $V(\phi_+) = 0$, where the limiting value $\phi_+$ is some finite constant.  
From our assumption in Section \ref{mmDEC} that $V \geq 0$ uniformly, it then follows that $V^\prime(\phi_+) = 0$ as well.   
Lastly, we assume that the Hawking mass $m$ is uniformly bounded along $\mathcal{S}$.

We remark that the requirement that $\phi$ have a finite limit along $\mathcal{S}$ results in a slight loss of generality, since we are excluding possibilities such as $\phi \sim \log r$ with $\lim_{x \rightarrow \infty} V(x) = 0$. 
However, this drawback is outweighed by the advantage of being able to estimate $V(\phi)$ and $V^\prime(\phi)$ in terms of $\phi$ using the mean value theorem.

For definiteness let us assume that $\mathcal{Q} \subset [0, u_0] \times [v_0, \infty) \subset \mathbb{M}^{1+1}$, some $u_0, v_0 > 0$, and that 
in particular $v \rightarrow \infty$ along $\mathcal{S}$.
We abuse notation somewhat to allow $\overline{\mathcal{Q}}$ to include points along the ``ray" $[0, u_0] \times \{ \infty \}$, so as to be able to refer to points ``at infinity," e.g. spacelike infinity $i^0$,  future timelike infinity $i^+$, and future null infinity $\mathcal{I}^+$.
(We could achieve the same end more rigorously by requiring our embedding $\mathcal{Q} \hookrightarrow \mathbb{M}^{1+1}$ to be a conformal compactification, but the certain core elements of the proof of the theorem are clearer with $v$ scaled to have infinite range on $\mathcal{Q}$.)
Constant-$u$ curves are said to be \emph{outgoing} and constant-$v$ ones \emph{ingoing}.

%%%%%%%%%%%%%%%%%%%%%%%%%%%%%%%%%%%%%%%%%%%%%%%%%%%%%%%%%%%%%%%%%%%%%%%%%%%%%%%%%%%%%%%%%%%%%%%%%%%%%%%%%%%%%%%%%
%
%
%				`PHYSICALLY REASONABLE' BLACK HOLE DATA
%
%%%%%%%%%%%%%%%%%%%%%%%%%%%%%%%%%%%%%%%%%%%%%%%%%%%%%%%%%%%%%%%%%%%%%%%%%%%%%%%%%%%%%%%%%%%%%%%%%%%%%%%%%%%%%%%%%%

\subsection{Physically reasonable black hole data} \label{BHdata}
Finally we suppose that the asymptotically flat initial hypersurface $\mathcal{S}$ is equipped with physically reasonable black hole initial data.  
By this we mean three things:  one, that $\mathcal{S}$ should contain at least one (spherically symmetric) closed trapped surface; two, that $\mathcal{S}$ should \emph{not} contain any (spherically symmetric) weakly anti-trapped surfaces; and three, that the data should decay sufficiently rapidly toward the asymptotically flat end to insure that future null infinity $\mathcal{I}^+ \neq \emptyset$, where $\mathcal{I}^+$ is defined as in \cite{D2}.  
We introduce relevant definitions and discuss the first two of these requirements in greater detail in the next section.  
As for the third assumption, we do not address here the issue of precisely what rate of decay is sufficiently rapid to insure that future null infinity $\mathcal{I}^+ \neq \emptyset$, we simply assume our initial data has this property.  
Alternately, we could avoid all discussion of decay by taking $\phi$ and $\nabla\phi$ to be compactly supported on $\mathcal{S}$ and assuming $V(0) = 0$; Birkhoff's theorem \cite{HE} would then imply that our spacetime contains a region isometric to an exterior region of the Schwarzschild spacetime and would thus guarantee that $\mathcal{I}^+ \neq \emptyset$.

With these assumptions in place, it now follows from \cite{D2} that $\mathcal{Q}$ contains a black hole region.  
In particular, $\mathcal{I}^+$ is complete with future limit point $i^+$ and past limit point $i^0$, the black hole region is $\mathcal{B} := \mathcal{Q}\setminus J^-(\mathcal{I}^+)$, and its event horizon is $\mathcal{H} := \del\left( \mathcal{Q}\setminus J^-(\mathcal{I}^+) \right) \cap \mathcal{Q}$. 
(Here and henceforth, set boundaries and closures are to be taken with respect to the topology of $\mathbb{M}^{1+1}$ rather than the relative topology of $\mathcal{Q}$.)
See Figure \ref{setup1} for a representative Penrose diagram.  
The event horizon necessarily has the property that $\sup_{p \in \mathcal{H}} r(p) = r_+ < \infty$.

\begin{figure}[hbtp]
\begin{center}
{
\psfrag{trapped}{\small{trapped surface}}
\psfrag{H+}{\small{$\mathcal{H}$}}
\psfrag{S}{\small{$\mathcal{S}$}}
\psfrag{i0}{\small{$i^0$}}
\psfrag{i+}{\small{$i^+$}}
\psfrag{scri}{\small{$\mathcal{I}^+$}}
\psfrag{rtoinfty}{\small{$r \rightarrow \infty$}}
\psfrag{B}{\small{$\mathcal{B}$}}
\resizebox{3.0in}{!}{\includegraphics{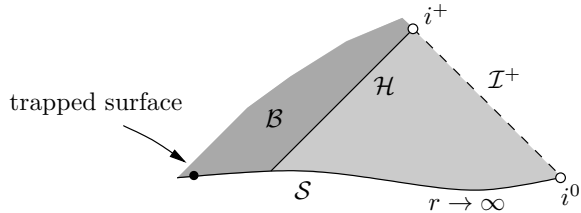}}
\hspace*{.5in}
}
\caption{\emph{A spacetime generated by ``physically reasonable asymptotically flat black hole initial data" as described in Sections \ref{AFevol} and \ref{BHdata}.
The existence of the black hole $\mathcal{B} = \mathcal{Q}\setminus J^-(\mathcal{I}^+)$ is guaranteed by \cite{D2}, given that $\mathcal{S}$ contains a trapped surface and that $\mathcal{I}^+ \neq \emptyset$; $\mathcal{H}$ is its event horizon.
}}
\label{setup1}
\end{center}
\end{figure}

%%%%%%%%%%%%%%%%%%%%%%%%%%%%%%%%%%%%%%%%%%%%%%%%%%%%%%%%%%%%%%%%%%%%%%%%%%%%%%%%%%%%%%%%%%%%%%%%%%%%%%%%%%%%%%%%%
%
%
%					REGULAR, TRAPPED REGIONS, & MTT
%
%%%%%%%%%%%%%%%%%%%%%%%%%%%%%%%%%%%%%%%%%%%%%%%%%%%%%%%%%%%%%%%%%%%%%%%%%%%%%%%%%%%%%%%%%%%%%%%%%%%%%%%%%%%%%%%%%%

\subsection{Regions $\mathcal{R}$, $\mathcal{T}$, and $\mathcal{A}$}
  
The characterization of a closed spacelike two-surface $S$ in $\mathcal{M}$ as trapped, marginally trapped, weakly anti-trapped, etc., depends on the signs of the inner and outer future null expansions, $\theta_+$ and $\theta_-$, on $S$.
These expansions are defined as follows: if $\ell^+$ and $\ell^-$ are two future-directed null vector fields normal to $S$, with $\ell^+$ pointing towards the asymptotically flat region and $\ell^-$ away from it, then at a point $p \in S$, 
\begin{equation*} \theta_\pm = \text{div}_S \ell^\pm = h^{\alpha \beta} \ell^\pm_{\beta; \alpha} \end{equation*}
where $h_{\alpha\beta}$ is the Riemannian metric on $S$ induced from $g$.
Rescaling $\ell^+$ or $\ell^-$ by a positive factor rescales the corresponding expansion but does not change its sign, so the signs of $\theta_{\pm}$ are well-defined given $\mathcal{M}$'s time-orientation and asymptotically flat end.  We then say that $S$ is \emph{trapped} if both $\theta_+ <0$ and $\theta_- < 0$ at all points in $S$, \emph{marginally trapped} if $\theta_+ = 0$ and $\theta_- < 0$ at all points in $S$, and \emph{weakly anti-trapped} if $\theta_- \geq 0$ at all points in $S$ (no restriction on $\theta_+$).  By contrast, closed spacelike two-surfaces in flat or nearly flat space will always have $\theta_+ > 0$ and $\theta_- < 0$ everywhere on $S$. 

Because of the spherical symmetry, $\theta_+$ and $\theta_-$ are constant on round two-spheres in $\mathcal{M}$ and therefore descend to pointwise functions on $\mathcal{Q}$.  Since $u$ is the ingoing and $v$ the outgoing coordinate on $\mathcal{Q}$, one computes that $\theta_+$ is proportional to $\del_v r$ and $\theta_-$ is proportional to $\del_u r$.  Thus the first two of the ``physically reasonable black hole initial data" assumptions listed above are satisfied provided that $\del_u r < 0$ everywhere on $\mathcal{S}$ and that there exists $q \in \mathcal{S}$ such that $\del_v r (q) < 0$ as well.  We note also that since $r$ tends to infinity in the direction of increasing $v$, $\del_v r$ must eventually become positive along $\mathcal{S}$ in the direction of the flat end, and there must therefore exist an outermost marginally trapped surface $p_\ast$ along $\mathcal{S}$.  (Here and henceforth all trapped or marginally trapped surfaces should be understood to be spherically symmetric, whether or not this is made explicit.)  This surface $p_\ast$ plays a large role in the statement and proof of the main theorem.

Now, the object of this paper is to locate and elucidate the causal and asymptotic behavior of any (spherically symmetric) \emph{marginally trapped tubes} in the spacetime.  In general, a marginally trapped tube in the spacetime $\mathcal{M}$ is a smooth hypersurface, of any causal character, which is foliated by closed, marginally trapped two-surfaces.  In our spherically symmetric setting, we restrict our attention to those marginally trapped tubes foliated by round two-spheres, in which case the tubes descend to curves in the quotient $\mathcal{Q}$. 
In fact, we define three subsets of interest in $\mathcal{Q}$: 
the \textit{regular region} 
\[ \mathcal{R} = \{ (u,v) \in \mathcal{Q} : \del_v r > 0 \text{ and } \del_u r < 0 \}, \] the \textit{trapped region} 
\[ \mathcal{T} = \{ (u,v) \in \mathcal{Q} : \del_v r < 0 \text{ and } \del_u r < 0 \}, \] and the \textit{marginally trapped tube}, 
\[ \mathcal{A} = \{ (u,v) \in \mathcal{Q} : \del_v r = 0 \text{ and } \del_u r < 0 \}. \] 
Note that $\mathcal{A}$ is a smooth hypersurface in $\mathcal{Q}$ wherever $0$ is a regular value of $\del_v r$.

It is clear that $\mathcal{Q} = \mathcal{R} \cup \mathcal{A} \cup \mathcal{T}$ if and only if $\mathcal{Q}$ contains no weakly anti-trapped surfaces.  We have assumed that $\mathcal{S}$ contains no such surfaces, and it turns out that this is sufficient to guarantee that none evolve:
\begin{prop}\cite{C,D2} \label{1} If $\del_u r < 0$ along $\mathcal{S}$, then $\del_u r < 0$ everywhere in $\mathcal{Q}$.
\end{prop}
\begin{proof}
Let $(u, v)$ be any point in $\mathcal{Q}$.  Suppose the ingoing null ray to the past of $(u,v)$ intersects $\mathcal{S}$ at the point $(u^\prime, v)$.  Then integrating Raychaudhuri's equation (\ref{eq1}) along this ray, we obtain
\[ (\Omega^{-2} \del_u r)(u, v) 
= (\Omega^{-2} \del_u r)(u^\prime, v) -  \int_{u^\prime}^{u} r \,\Omega^{-2} (\del_u \phi)^2 (\bar{u}, v) \, d\bar{u}. \]
By assumption $\del_u r (u^\prime, v) < 0 $, so the right-hand side of this equation is strictly negative, and hence so is the left-hand side.
\end{proof}

Integrating the other Raychaudhuri equation (\ref{eq2}) yields a slightly different but equally useful result:
\begin{prop}\cite{C,D2} \label{handy}  If $(u,v) \in \mathcal{T} \cup \mathcal{A}$, then $(u, v^\prime) \in \mathcal{T} \cup \mathcal{A}$  for all $v^\prime > v$.  Similarly, if $(u,v) \in \mathcal{T}$, then $(u, v^\prime) \in \mathcal{T}$  for all $v^\prime > v$.
\end{prop}
\begin{proof} Integrating (\ref{eq2}) along the null ray to the future of a point $(u,v) \in \mathcal{Q}$ yields
\[ (\Omega^{-2} \del_v r)(u, v^\prime) = (\Omega^{-2} \del_v r)(u, v) - \int_{v}^{v^\prime} r \,\Omega^{-2} (\del_v \phi)^2 (u, \bar{v}) \,d\bar{v} \]
for $v^\prime > v$.  The right-hand side of this equation is nonpositive if $\del_v r(u, v) \leq 0$, and strictly negative if
$\del_v r(u, v) < 0$; both statements of the proposition follow. 
\end{proof}

Since $\mathcal{I}^+$ is characterized by the property that $r$ has infinite supremum along any outgoing null ray with a limit point on $\mathcal{I}^+$, an immediate consequence of Proposition \ref{handy} is that all trapped and marginally trapped surfaces must lie inside the black hole, $\mathcal{T} \cup \mathcal{A} \subset \mathcal{B}$.

We shall be concerned only with the connected component of $\mathcal{R}$ containing the exterior of the black hole $J^-(\mathcal{I}^+)$, so let us assume for convenience that $\mathcal{R} \cap \mathcal{S}$ is connected, i.e.\ that points on $\mathcal{S}$ interior to $p_\ast$ are trapped or marginally trapped.  
Indeed, since we made no assumptions concerning the inner boundary of $\mathcal{S}$, we can cut off any other components of $\mathcal{S} \cap \mathcal{R}$ without affecting our hypotheses.  
Proposition \ref{handy} now guarantees that $\mathcal{R}$ is connected in $\mathcal{Q}$.  
Figure \ref{setup2} provides a Penrose diagram indicating a possible configuration of $\mathcal{R}$, $\mathcal{T}$, and $\mathcal{A}$ in $\mathcal{Q}$.

\begin{figure}[hbtp]
\begin{center} 
{
\psfrag{A}{\scriptsize{$\mathcal{A}$}}
\psfrag{H+}{\scriptsize{$\mathcal{H}$}}
\psfrag{S}{\scriptsize{$\mathcal{S}$}}
\psfrag{T}{\scriptsize{$\mathcal{T}$}}
\psfrag{R}{\scriptsize{$\mathcal{R}$}}
\psfrag{i0}{\scriptsize{$i^0$}}
\psfrag{i+}{\scriptsize{$i^+$}}
\psfrag{scri}{\scriptsize{$\mathcal{I}^+$}}
\psfrag{p*}{\scriptsize{$p_\ast$}}
\resizebox{2.75in}{!}{\includegraphics{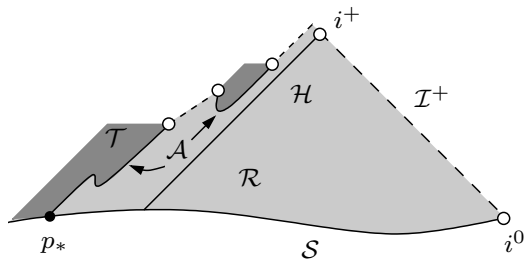}}
}
\caption{\emph{The regions $\mathcal{R}$ and $\mathcal{T}$, shaded light and dark, respectively, are separated by marginally trapped tube $\mathcal{A}$. 
The point $p_\ast$ represents the outermost marginally trapped sphere along $\mathcal{S}$; the dashed lines inside the black hole indicate portions of the spacetime boundary, $\overline{\mathcal{Q}}\setminus \mathcal{Q}$.
As depicted here, $\mathcal{A}$ need be neither achronal nor connected, but it must lie inside the black hole and comply with Proposition \ref{handy}.}}
\label{setup2}
\end{center}
\end{figure}

%%%%%%%%%%%%%%%%%%%%%%%%%%%%%%%%%%%%%%%%%%%%%%%%%%%%%%%%%%%%%%%%%%%%%%%%%%%%%%%%%%%%%%%%%%%%%%%%%%%%%%%%%%%%%%%%%%%%
%
%						PROOF OUTLINE
%
%%%%%%%%%%%%%%%%%%%%%%%%%%%%%%%%%%%%%%%%%%%%%%%%%%%%%%%%%%%%%%%%%%%%%%%%%%%%%%%%%%%%%%%%%%%%%%%%%%%%%%%%%%%%%%%%%%%%

With notation in place, we now give a brief overview of the proof of the main theorem.
Logically speaking, the proof comprises three main parts, though they are presented somewhat out of order in the actual proof given in Section \ref{main}.  
The first and most involved piece is a bootstrap argument whose purpose is to establish positive lower bounds for two particular quantities in the region $\overline{\mathcal{R}} \cap \{ r \leq R \}$ for a carefully chosen constant $R$.
The two quantities in question are called $\kappa$ and $\alpha$ (defined by \eqref{kappadef} and \eqref{alphadef}, respectively), and the constant $R$ is chosen in such a way that $\mathcal{R} \cap \{ r < R \}$ is an open neighborhood of $\mathcal{H}$ but $R$ is not too large.
The primary tools used in closing the bootstrap are the smallness of initial data stipulated by the theorem and the energy estimates from Section \ref{MandEE} below. 
The second part of the proof shows that the component of $\mathcal{A}$ containing $p_\ast$ must be achronal and that it must extend all the way out to $v = \infty$, i.e.\ either to $i^+$ or the Cauchy horizon.
The key ingredients here are the lower bound for $\alpha$ obtained from the bootstrap and the extension principle of \cite{D2}, which is known to hold for the Higgs field matter model. 
The last step of the proof is to show that if $\mathcal{A}$ terminates at the Cauchy horizon, leaving a gap between itself and $\mathcal{H}$, then we can derive a contradiction essentially by integrating  $\del^2_{uv} r$ twice in the region formed by this gap.
The contradiction arises directly from the lower bounds on $\alpha$ and $\kappa$ obtained from the bootstrap argument in addition to bounds on the radial function $r$.

%%%%%%%%%%%%%%%%%%%%%%%%%%%%%%%%%%%%%%%%%%%%%%%%%%%%%%%%%%%%%%%%%%%%%%%%%%%%%%%%%%%%%%%%%%%%%%%%%%%%%%%%%%%%%%%%%
%
%
%					MONOTONICITY & ENERGY ESTIMATES
%
%%%%%%%%%%%%%%%%%%%%%%%%%%%%%%%%%%%%%%%%%%%%%%%%%%%%%%%%%%%%%%%%%%%%%%%%%%%%%%%%%%%%%%%%%%%%%%%%%%%%%%%%%%%%%%%%%%

\subsection{Monotonicity \& energy estimates} \label{MandEE}

From Proposition \ref{1} and our assumption of no weakly anti-trapped surfaces on $\mathcal{S}$, we have
\begin{equation} \del_u r < 0 \quad \text{ everywhere in } \mathcal{Q},  \label{duneg}
\end{equation}
and by definition,
\begin{equation} \del_v r \geq 0 \quad \text{ everywhere in } \mathcal{R} \cup \mathcal{A}. \label{dvpos}
\end{equation}
These inequalities, together with equations \eqref{eq3}, \eqref{eq4}, and the dominant energy condition $V(x) \geq 0$, imply that 
everywhere in $\mathcal{R} \cup \mathcal{A}$ we have
\begin{equation} \del_u m \leq 0  \label{dudm}
\end{equation}
and
\begin{equation} \del_v m \geq 0. \label{dvdm}
\end{equation}
Set 
\begin{equation} M := \sup_{q\in\mathcal{S}} m(q) \label{M},
\end{equation}
\begin{equation} m_1 := m(p_\ast), \label{m1} \end{equation}
and
\begin{equation} r_1 := r(p_\ast), \label{r1} \end{equation}
where $p_\ast$ corresponds to the outermost marginally trapped sphere along $\mathcal{S}$ as described in the previous section.  
Since $\mathcal{S}$ is spacelike, \eqref{duneg}, \eqref{dvpos} and \eqref{dudm}, \eqref{dvdm} imply that $r$ is increasing and $m$ is non-decreasing, respectively, along $\mathcal{S} \cap \mathcal{R}$ toward the asymptotically flat end.  Thus we have
\begin{equation} m_1 = \inf_{q\in\mathcal{S}\cap \overline{\mathcal{R}}} m(q) \label{m1inf}
\end{equation}
and
\begin{equation} r_1 = \inf_{q\in\mathcal{S}\cap \overline{\mathcal{R}}} r(q). \label{r1inf}
\end{equation}
Note also that $M < \infty$ by hypothesis (Section \ref{AFevol}), and $m_1, r_1 > 0$, since \eqref{mdef} and the fact that $\del_v r (p_\ast) = 0$ together yield $r_1 = 2 m_1$, and we assumed in Section \ref{AFevol} that $r > 0$ on $\mathcal{S}$. 

\begin{prop} \label{mbounds} For all $q \in \overline{\mathcal{R}}$, 
\begin{equation} m_1 \leq m(q) \leq M \label{mboundseq}\end{equation}  
and
\begin{equation} r_1 \leq r(q). \label{rboundseq}\end{equation}
\end{prop}
\begin{proof} Given $q \in \overline{\mathcal{R}}$, the outgoing null ray to the past of $q$  must also lie in $\overline{\mathcal{R}}$ by Proposition \ref{handy}.  Then \eqref{rboundseq} and the left-hand inequality of \eqref{mboundseq} follow immediately from \eqref{dvdm} and \eqref{m1inf}, or \eqref{dvpos} and \eqref{r1inf}, respectively.   

The ingoing null ray to the past of a point $q \in \overline{\mathcal{R}}$ may or may not lie entirely in $\overline{\mathcal{R}}$.  If it does, then the right-hand inequality of \eqref{mboundseq} is immediate from \eqref{dudm} and \eqref{M}.  If not, then suppose $q = (u^\prime, v^\prime)$, so that the ingoing null ray in question is given by $\{ (u, v) : v = v^\prime \}$, and suppose $(u^{\prime\prime}, v^\prime)$ is the point at which this ray intersects $\mathcal{S}$.
Further, let $u_1$ and $u_2$ be the minimum and maximum, respectively, of the set $\{ u \in [u^{\prime\prime}, u^\prime] : (u, v^\prime) \in \mathcal{T} \cup \mathcal{A} \}$.  It follows that $(u_1, v^\prime), (u_2, v^\prime) \in \mathcal{A}$, so  $\del_v r (u_1, v^\prime) = \del_v r (u_2, v^\prime) = 0$. Then applying inequalities \eqref{duneg}, \eqref{dudm}, and equation \eqref{mdef}, we have
\[ m(q)  \leq m(u_2, v^\prime) = \frac{r}{2}(u_2, v^\prime) \leq \frac{r}{2}(u_1, v^\prime) = m(u_1, v^\prime) \leq m(u^{\prime\prime}, v^\prime) \leq M, \]
so the right-hand inequality holds as well.
\end{proof}

Let $\kappa$ denote the quantity
\begin{equation} \kappa := -\frac{\Omega^2}{4 \del_u r} = \frac{\del_v r}{1 - \frac{2m}{r} } \label{kappadef}\end{equation}
and observe that \eqref{duneg} implies that $\kappa > 0$ everywhere in $\mathcal{Q}$.

We now have the necessary tools to derive the two energy estimates for $\phi$ in $\overline{\mathcal{R}}$ which are crucial for the main theorem.  
\begin{prop} For any interval $[ u_1, u_2 ] \times \{ v \} \subset \overline{\mathcal{R}}$,
\begin{equation} - \int_{u_1}^{u_2} \textstyle\frac{1}{2} \left(1 - \frac{2m}{r}\right)\left( \frac{r \del_u \phi}{\del_u r} \right)^2 \del_u r  (\bar{u}, v) \, d\bar{u} \, \leq \, M - m_1, \label{u-energy}\end{equation}
and for any interval $\{ u \} \times [v_1 , v_2] \subset \overline{\mathcal{R}}$,
\begin{equation} \int_{v_1}^{v_2} \textstyle\frac{1}{2} r^2 \kappa^{-1} (\del_v \phi)^2 (u,\bar{v}) d\bar{v} \leq M - m_1. \label{v-energy}
\end{equation}
\end{prop}
\begin{proof}
To obtain the first of these estimates, we integrate equation \eqref{eq3} along  $[ u_1, u_2 ] \times \{ v \}$ and apply \eqref{duneg}, the dominant energy condition, and Proposition \ref{mbounds} to get
\begin{eqnarray*} m_1 - M & \leq & 
\int_{u_1}^{u_2}   r^2 \left( V(\phi) \del_u r - 2 \Omega^{-2} (\del_u\phi)^2 \del_v r \right)  (\bar{u}, v) d\bar{u} \\
& \leq & 
\int_{u_1}^{u_2} - \left( 2 r^2 \Omega^{-2} (\del_u\phi)^2 \del_v r \right) (\bar{u}, v) d\bar{u}.
\end{eqnarray*}
Using \eqref{kappadef} and rearranging the factors of the integrand, this yields \eqref{u-energy}.

Analogously, for the second estimate, we integrate equation \eqref{eq4} along $\{ u \} \times [v_1 , v_2]$  and apply \eqref{dvpos}, the dominant energy condition, and Proposition \ref{mbounds}, arriving at 
\begin{eqnarray*}
M - m_1 & \geq & \int_{v_1}^{v_2} r^2 \left( V(\phi) \del_v r - 2 \Omega^{-2} (\del_v\phi)^2 \del_u r \right) (u, \bar{v}) d\bar{v} \\
& \geq & \int_{v_1}^{v_2} - \left( 2 r^2\Omega^{-2} (\del_v\phi)^2 \del_u r \right) (u, \bar{v}) d\bar{v}.
\end{eqnarray*}
Again, making use of \eqref{kappadef} and rearranging  yields \eqref{v-energy}.
\end{proof}

%%%%%%%%%%%%%%%%%%%%%%%%%%%%%%%%%%%%%%%%%%%%%%%%%%%%%%%%%%%%%%%%%%%%%%%%%%%%%%%%%%%%%%%%%%%%%%%%%%%%%%%%%%%%%%%%%%%%%%%%%%%%%%%%%%%%%%%
%
%						THEOREM
%
%%%%%%%%%%%%%%%%%%%%%%%%%%%%%%%%%%%%%%%%%%%%%%%%%%%%%%%%%%%%%%%%%%%%%%%%%%%%%%%%%%%%%%%%%%%%%%%%%%%%%%%%%%%%%%%%%%%%%%%%%%%%%%%%%%%%%%%

\section{Main Result} \label{main}

\begin{theorem} 
Consider spacelike spherically symmetric initial data for the Einstein-Higgs field equations with one asymptotically flat end.  
Assume that the Higgs potential function is such that the dominant energy condition is satisfied, i.e.\ 
\[ V(x) \geq 0 \,\,\text{ for all } x\in \mathbb{R}, \] 
and assume further that 
\[ |V^{\prime\prime}(x)| \leq B < \infty \,\,\text{ for all } x\in \mathbb{R}. \] 
As described in Section \ref{prelim}, suppose $(\mathcal{Q}, \Omega, r, \phi)$ is the 2-dimensional Lorentzian quotient of the future Cauchy development of the initial hypersurface $\mathcal{S}$ equipped with asymptotically flat physically reasonable black hole initial data, with
$\phi \rightarrow \phi_+$ along $\mathcal{S}$ for some $\phi_+ \in \mathbb{R}$, and with global coordinates $(u, v)$ obtained from identifying $\mathcal{Q}$ with its image under a conformal embedding into $\mathbb{M}^{1+1}$.
Let $p_\ast \in \mathcal{Q}$ be the point corresponding to the outermost marginally trapped sphere along $\mathcal{S}$, and let $\mathcal{H}$ denote the event horizon of the black hole region in $\mathcal{Q}$.

\smallskip

Define $m_1$, $M$, and $\kappa$ as in \eqref{m1}, \eqref{M} and \eqref{kappadef}, respectively, and assume that there exists a constant $\kappa_0$ such that
\begin{equation} 0 < \kappa_0 \leq \kappa(q) \leq 1 \label{kappa0def} \end{equation}
for all $q\in \mathcal{S}$.  
Fix a constant $C$ such that
\begin{equation} C > \sup \left\{ \left|  \frac{r \del_u \phi}{\del_u r} (q) \right| \, : \, q \in \mathcal{S},  r(q) \leq 2M \right\}. 
\label{Cdef}
\end{equation}
Then there exist constants $\epsconst > 0$ and $\rhoconst \in (0,1)$, depending only on $M$, $B$, $\kappa_0$, and $C$, such that if
\begin{equation} |\phi - \phi_+| \leq \frac{\epsconst}{\sqrt{r}} \quad \text{ on } \mathcal{S} \label{phicond}\end{equation}
and
\begin{equation} \rhoconst \leq \frac{m_1}{M},  \label{mcond} \end{equation} 
then $\mathcal{Q}$ contains a connected, achronal, marginally trapped tube $\mathcal{A}$ which intersects $\mathcal{S}$ at $p_\ast$ and is asymptotic to the event horizon $\mathcal{H}$.  
See Figure \ref{thm} for a representative Penrose diagram.
\end{theorem}
\begin{figure}[hbtp]
\begin{center}
{
\psfrag{A}{\small{$\mathcal{A}$}}
\psfrag{H+}{\small{$\mathcal{H}$}}
\psfrag{S}{\small{$\mathcal{S}$}}
\psfrag{i0}{\small{$i^0$}}
\psfrag{i+}{\small{$i^+$}}
\psfrag{scri}{\small{$\mathcal{I}^+$}}
\psfrag{p*}{\small{$p_\ast$}}
\resizebox{2.25in}{!}{\includegraphics{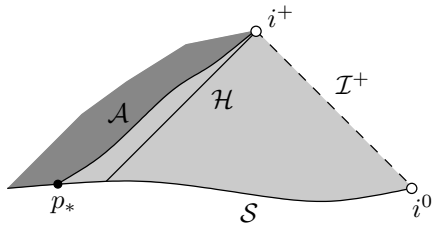}}
}
\caption{\emph{The result of the theorem is that the marginally trapped tube $\mathcal{A}$ emanating from the outermost marginally trapped sphere $p_\ast$ is in fact connected, achronal, and asymptotic to the event horizon $\mathcal{H}$, as shown here.}}
\label{thm}
\end{center}
\end{figure}
%
%
%
%%%%%%%%%%%%%%%%%%%%%%%%%%%%%%%%%%%%%%%%%%%%%%%%%%%%%%%%%%%%%%%%%%%%%%%%%%%%%%%%%%%%%%%%%%%%%%%%%%%%%%%%%%%%%%%%%
%
%
%			 			REMARKS
%
%%%%%%%%%%%%%%%%%%%%%%%%%%%%%%%%%%%%%%%%%%%%%%%%%%%%%%%%%%%%%%%%%%%%%%%%%%%%%%%%%%%%%%%%%%%%%%%%%%%%%%%%%%%%%%%%%%
%
%
%
\noindent 
\emph{Remarks.} 
Assumption \eqref{kappa0def} amounts to fixing a gauge along $\mathcal{S}$, while \eqref{Cdef} poses no restriction on the initial data.

The constants $\epsconst$ and $\rhoconst$ should both be understood as smallness parameters.  
In particular, $\rhoconst$ should be thought of as being very close to 1, forcing $m_1$ to be close to $M$ and hence allowing very little matter flux along the asymptotically flat end.

The possibility that $\mathcal{A}$ coincides which with part or all of $\mathcal{H}$ is included in our interpretation of the term  ``asymptotic."
Indeed, there are two different, reasonable definitions of what it means for $\mathcal{A}$ to be ``asymptotic to" $\mathcal{H}$. 
We discuss these in the course of the proof and show that $\mathcal{A}$ is asymptotic to $\mathcal{H}$ in both senses.

Finally, although we do not pursue the matter in the present paper, it seems plausible that one could remove the requirement that $V$ have a uniformly bounded second derivative on all of $\mathbb{R}$. 
Instead, one would carefully track the range of $\phi$, showing that it lies in some compact interval which can be identified \emph{a priori}, depending only on the initial data, and then replace the uniform bound $B$ with the maximum of $|V^{\prime\prime}|$ over this interval.

%%%%%%%%%%%%%%%%%%%%%%%%%%%%%%%%%%%%%%%%%%%%%%%%%%%%%%%%%%%%%%%%%%%%%%%%%%%%%%%%%%%%%%%%%%%%%%%%%%%%%%%%%%%%%%%%%
%
%
%					        PROOF
%
%%%%%%%%%%%%%%%%%%%%%%%%%%%%%%%%%%%%%%%%%%%%%%%%%%%%%%%%%%%%%%%%%%%%%%%%%%%%%%%%%%%%%%%%%%%%%%%%%%%%%%%%%%%%%%%%%%

\begin{proof}
To begin, note that since \eqref{eq1}-\eqref{sfequv} depend only on derivatives of $\phi$ except through the potential $V(\phi)$, by translating $\phi$ and shifting the domain of $V$, we may assume without loss of generality that $\phi_+ = 0$ and $V(0) = V^\prime(0) = 0$ (cf.\ Section \ref{AFevol}).
Such a modification has no real effect except to simplify the presentation of proof.

The proof comprises four parts. 
First we establish the necessary constants, including $\epsconst$ and $\rhoconst$ and a host of auxiliary parameters.  
Second, we set up a bootstrap region $\mathcal{V}$ in $\mathcal{R}$ in which certain nice estimates hold.  
Third, we show how the statement of the theorem follows from retrieval of the bootstrap estimates in $\overline{\mathcal{V}} \cap \mathcal{R}$.  
And finally, we show that $\mathcal{V} \neq \emptyset$ and retrieve each of the bootstrap conditions in $\overline{\mathcal{V}} \cap \mathcal{R}$.

\medskip
First fix $\epsconst$ to be any value such that
\begin{equation}  \epsconst < \frac{1}{ \sqrt{8 M B}}. 
\end{equation}
Prescribing $\rhoconst$ is less straightforward, since its dependence on $M$, $B$, $\kappa_0$ and $C$ is rather complicated. 
Our definition below does not necessarily produce the optimal (sharp) value of $\rhoconst$.

In order to determine $\rhoconst$, we first define auxiliary constants $\sigma$, $L$, $\Lambda$, $k$, and $\lambda$.  
First choose $\sigma < 1$ sufficiently close to $1$ that
\begin{equation} \frac{\epsconst}{\sigma} < \frac{1}{ \sqrt{8 M B}}, \label{sigmadef}
\end{equation}
and then choose $L > 0$ sufficiently small that
\begin{equation} \epsconst + \sqrt{2LM} < \frac{\epsconst}{\sigma}. \label{Ldef}
\end{equation}
Set
\begin{equation} \Lambda := C + \frac{9Me^L\sqrt{B}}{2 \sigma^2 \kappa_0} + \frac{3e^{L/2}}{\sigma}\sqrt{\frac{L}{2\kappa_0}}. \label{Lambdadef}
\end{equation}
Next choose $k > 1$ sufficiently close to $1$ that 
\begin{equation} \log{k} < 
\min\left\{ \textstyle\frac{1}{2} \Lambda^{-2}
\log 2, \,
\log\left(\textstyle \frac{3}{2}\right) \right\} \label{kdef}
\end{equation}
\it and \rm the following inequality is satisfied:
\begin{equation} k^3 \left( \frac{\epsconst}{\sigma \sqrt{2M}} +  \sqrt{\textstyle\frac{1}{2}\log 2 \log k } \right)^2 < \frac{1}{16 M^2 B}.
\label{kdef2}
\end{equation}
Note that such a choice is possible by \eqref{sigmadef}.
Finally, choose $\lambda$ such that
\begin{equation} 2 < \lambda < 2k  \label{lambdadef}
\end{equation}
and, recalling \eqref{Cdef}, such that
\begin{equation} \sup \left\{ \left|  \frac{r \del_u \phi}{\del_u r} (q) \right| \, : \, q \in \mathcal{S},  r(q) \leq \lambda M \right\} \leq C,
\label{lambdadef2}
\end{equation}
and note that \eqref{kdef} and \eqref{lambdadef} imply that $\lambda < 3$.
We now fix $\rhoconst$ such that 
\begin{equation} \max\left\{ 
\frac{\lambda}{2k}, \,
1 - L\left( 1 - \frac{2}{\lambda} \right), \,
\sigma^2 
\right\}
< \rhoconst
< 1.
\label{rhodef}
\end{equation}
Such a choice is possible because \eqref{lambdadef} and the fact that $\sigma < 1$ together imply that the maximum on the left-hand side of \eqref{rhodef} is strictly less than $1$.
It is helpful to note that \eqref{mcond} now implies not only that $m_1 \geq \rhoconst M$, but also $r_1 \geq 2\rhoconst M$, since $2m_1 = 2m(p_\ast) = r(p_\ast) =  r_1$.  
Two other useful consequences of this definition \eqref{rhodef} are that $\frac{\lambda}{2\rhoconst} < k$ and $1 - \rhoconst < L$.

\medskip

Next we define all the constants needed for the bootstrap argument.
First, for convenience, set
\begin{eqnarray}
R & := &  \lambda M \label{Rdef} \\
\ell & := & \frac{1-\rhoconst}{1-\frac{2}{\lambda}}. \label{elldef}
\end{eqnarray}
Note that \eqref{rhodef} implies that $\ell < L$.
Finally, define
\begin{eqnarray}
C_1 & := & \sqrt{\frac{\log 2}{2 \log (\frac{\lambda}{2\rhoconst}) }} \label{C1} \\
\widehat{C} & := & \frac{\epsconst}{\sigma \sqrt{2M}} + \sqrt{\textstyle\frac{1}{2}\log 2 \log k} \label{Chat} \\
%\widetilde{C} & := & \frac{\epsconst}{\sigma} \sqrt{\frac{1}{ 2M } } \label{Ctilde}\\
\overline{C} & := & \frac{\sqrt{B}}{4M} \label{Coverline} \\
%\overline{\kappa_0} & := & \kappa_0 \cdot e^{-2\ell/\lambda } \label{kappa0bar}\\
\kappa_1 & := & \textstyle\frac{1}{2}\kappa_0 \cdot e^{-2\ell/\lambda } \label{kappa1}\\
\alpha_1 & := & \textstyle\frac{1}{2} \rhoconst M. \label{alpha1}
\end{eqnarray}

\medskip
%%%%%%%%%%%%%%%%%%%%%%%%%%%%%%%%%%%%%%%%%%%%%%%%%%%%%%%%%%%%%%%%%%%%%%%%%%%%%%%%%%%%%%%%%%%%%%%%%%%%%%%%%%%%
%
%					THE BOOTSTRAP REGION
%
%%%%%%%%%%%%%%%%%%%%%%%%%%%%%%%%%%%%%%%%%%%%%%%%%%%%%%%%%%%%%%%%%%%%%%%%%%%%%%%%%%%%%%%%%%%%%%%%%%%%%%%%%%%%

Now we have the necessary constants in place to define our bootstrap region.
Let $\mathcal{V}$ be the set of all points $(u,v)$ in $\{r \leq R \} \cap \mathcal{Q}$ such that the following six inequalities are satisfied for all $(\tilde{u}, \tilde{v}) \in J^{-}(u,v) \cap \{ r \leq R \}$: 
\begin{align}  \left| \frac{r \del_u \phi}{\del_u r} (\tilde{u}, \tilde{v} )\right|  \quad < & \quad C_1 \label{redshift}\\
|\phi(\tilde{u}, \tilde{v}))|  \quad < &  \quad \widehat{C} \label{phibound} \\
|V^\prime(\phi(\tilde{u}, \tilde{v}))|  \quad < &  \quad \overline{C} \label{potential}\\
\alpha(\tilde{u}, \tilde{v}) \quad > & \quad \alpha_1 \label{alpha}\\
\kappa(\tilde{u}, \tilde{v}) \quad > & \quad \kappa_1 \label{kappa}\\
\del_v r (\tilde{u}, \tilde{v})  \quad > & \quad 0, \label{drdv}
\end{align}
where the function $\alpha(u,v)$ is defined as
\begin{equation}\label{alphadef} \alpha := m - r^3 V(\phi).
\end{equation}

\medskip

%%%%%%%%%%%%%%%%%%%%%%%%%%%%%%%%%%%%%%%%%%%%%%%%%%%%%%%%%%%%%%%%%%%%%%%%%%%%%%%%%%%%%%%%%%%%%%%%%%%%%%%%%%%%%%%%%
%
%
%					REMAINDER OF PROOF ASSUMING BOOTSTRAP
%
%%%%%%%%%%%%%%%%%%%%%%%%%%%%%%%%%%%%%%%%%%%%%%%%%%%%%%%%%%%%%%%%%%%%%%%%%%%%%%%%%%%%%%%%%%%%%%%%%%%%%%%%%%%%%%%%%%

The first step of the proof is to show that inequalities \eqref{redshift}-\eqref{drdv} hold along the curve $\mathcal{S}_o$, where
$\mathcal{S}_o := \mathcal{S} \cap \mathcal{R} \cap \{ r \leq R \} = \mathcal{S} \cap \{ r_1 < r \leq R \}$; then $\mathcal{V}$ contains an open neighborhood of $\mathcal{S}_o$ in $\mathcal{Q}$.  
Since $\mathcal{V}$ is a past set in $\{ r \leq R \}$, i.e.\ $J^-(\mathcal{V}) \cap \{ r \leq R \} \subset \mathcal{V}$, its future boundary $\del^+ \mathcal{V} := \overline{\mathcal{V}}\setminus\mathcal{V}$ must be achronal.  
(Recall that set closures are taken here with respect to the topology of the underlying Minkowski space, so in particular, $\del^+ \mathcal{V}$ need not  \emph{a priori} lie entirely in $\mathcal{Q}$.)
 
Next consider a point $p \in \overline{\mathcal{V}} \cap \mathcal{R}$. By definition of $\mathcal{V}$, $p$ has the property that non-strict versions of inequalities \eqref{redshift}-\eqref{drdv} hold for all $q \in J^-(p) \cap  \{r \leq R \}$.  
Using this property, we shall show below that in fact \emph{strict} inequalities \eqref{redshift}-\eqref{drdv} hold at $p$.  
Since strict inequalities \eqref{redshift}-\eqref{drdv} must then also hold in a neighborhood of $p$ and $\del^+\mathcal{V}$ is achronal, we must have $p \in \mathcal{V}$.  
That is, $\mathcal{V}$ is closed in $\mathcal{R}$, and hence in $\mathcal{R} \cap \{ r \leq R \}$ as well.  
Clearly $\mathcal{V}$ is also open in $\mathcal{R} \cap \{ r \leq R \}$.  
Thus by continuity we have $\mathcal{V} = \mathcal{R} \cap \{ r \leq R \}$, and it follows that $\del^+ \mathcal{V} \cap \mathcal{Q} \subset \mathcal{A}$.  

Setting aside for the moment the matter of retrieving the bootstrap inequalities in $\overline{\mathcal{V}} \cap \mathcal{R}$, we now show how the statement of the theorem follows from these assertions, i.e.\ that $\mathcal{V} = \mathcal{R} \cap \{ r \leq R \}$ and $\del^+ \mathcal{V} \cap \mathcal{Q} \subset \mathcal{A}$.

It was shown in \cite{Wil} that if the inequality $T_{uv} \Omega^{-2} < \frac{1}{4r^2}$ holds at every point of $\mathcal{A}$, then each connected component of $\mathcal{A}$ is a smooth curve in $\mathcal{Q}$ which is everywhere spacelike or outgoing-null.  
Since $T_{uv} = \frac{1}{2} \Omega^2 V(\phi)$ here, the inequality becomes $2r^2 V(\phi) < 1$ in our setting.  
Now from inequality \eqref{alpha}, we have that $\alpha \geq \alpha_1 > 0$ everywhere in $\overline{\mathcal{V}}$, so in particular at points in $\del^+\mathcal{V} \cap \mathcal{Q} \subset \mathcal{A}$ we have
\[ 0 < \alpha = m - r^3 V(\phi) = \frac{r}{2}- r^3 V(\phi)  = \frac{r}{2}\left( 1 - 2r^2 V(\phi) \right). \]
Thus each connected component of $\del^+\mathcal{V} \cap \mathcal{Q}$ is indeed a smooth curve which is everywhere spacelike or outgoing-null.
 
Consider the connected component of $\del^+\mathcal{V} \cap \mathcal{Q}$ containing the outermost marginally trapped sphere $p_\ast$.
Then $p_\ast$ must constitute the inner endpoint of this curve segment; let $q_\ast = (u_\ast, v_\ast)$ denote the segment's other, outer endpoint.  
As described in Section \ref{AFevol}, we abuse notation slightly to allow the possibility that $v_\ast = \infty$.
Since $\del^+\mathcal{V} \cap \mathcal{Q}$ is smooth and consequently non-degenerate at every point, $q_\ast \notin \mathcal{Q}$.
Thus $q_\ast \in \overline{\mathcal{Q}}\setminus \mathcal{Q}$, and either $v_\ast < \infty$ or $v_\ast = \infty$.  

We wish to show that $v_\ast$ must be infinite, so suppose by way of contradiction that $v_\ast < \infty$. 
Since $\mathcal{V}$ is a past set in $\{ r \leq R \}$ and $q_\ast \in \overline{\mathcal{V}}$, we must have $I^-(q_\ast) \cap J^+(\mathcal{S}) \subset \mathcal{V}  \cup \{ r \geq R \} \subset \mathcal{R}$.
Let $u^\prime$, $v^\prime$ be such that $(u_\ast, v^\prime), (u^\prime, v_\ast) \in \mathcal{S}$.
Now, $q_\ast$ is the outer endpoint of a smooth, nonempty, achronal curve in $\mathcal{Q}$, so the global hyperbolicity of $\mathcal{Q}$ implies that the outgoing null ray to from $\mathcal{S}$ to $q_\ast$, i.e.\ $\{ u_\ast \} \times [v^\prime, v_\ast)$, must lie entirely in $\mathcal{Q}$ and hence in $\mathcal{R} \cup \mathcal{A}$. 
The ingoing null ray from $\mathcal{S}$ to $q_\ast$ need not lie entirely in $\mathcal{Q}$, however, so let 
\[u_{\ast\ast} = \sup \left\{ u \in [u^\prime, u_\ast] :  [u^\prime, u) \times \{ v_\ast \} \subset \mathcal{Q} \right\}, \]
and set $q_{\ast\ast} = (u_{\ast\ast}, v_\ast)$.
(Note that $q_\ast$ may equal $q_{\ast\ast}$.)
Then $I^-(q_{\ast\ast}) \cap  J^+(\mathcal{S}) \subset  I^-(q_\ast) \cap J^+(\mathcal{S}) \subset \mathcal{R}$, but since $\mathcal{Q}\setminus\mathcal{S}$ is open, $q_{\ast\ast} \notin \mathcal{Q}$.

Finally, fix a point $p \in \mathcal{S} \cap I^-(q_{\ast\ast})$.
By construction, $J^-(q_{\ast\ast}) \cap J^+(p) \setminus \{ q_{\ast\ast} \}$ lies in $\left(\overline{\mathcal{V}} \cup \{ r \geq R \}\right) \cap \mathcal{Q}$ and hence in $\mathcal{R} \cup \mathcal{A}$ as well.
(See Figure \ref{thmpf} for a Penrose diagram depicting these points and regions.)
But since the bound \eqref{rboundseq} of Proposition \ref{mbounds} guarantees that $q_{\ast\ast}$ lies away from the center of symmetry of $\mathcal{Q}$, the extension principle formulated in \cite{D2}, which was shown  in \cite{D3} to hold for self-gravitating Higgs fields with $V(x)$ bounded below, implies that $q_{\ast\ast} \in \mathcal{R} \cup \mathcal{A} \subset \mathcal{Q}$, a contradiction.  
So indeed we must have $v_\ast = \infty$.

\begin{figure}[hbtp]
\begin{center}
{
\psfrag{A}{\small{$\mathcal{A}$}}
\psfrag{T}{\small{$\mathcal{T}$}}
\psfrag{q*}{\small{$q_\ast$}}
\psfrag{q**}{\small{$q_{\ast\ast}$}}
\psfrag{S}{\small{$\mathcal{S}$}}
\psfrag{p}{\small{$p$}}
\psfrag{p*}{\small{$p_\ast$}}
\psfrag{p1}{\small{$(u_\ast, v^\prime)$}}
\psfrag{p3}{\small{$(u^\prime, v_\ast)$}}
\psfrag{u=u**}{\small{$u=u_{\ast\ast}$}}
\psfrag{int}{\small{$J^-(q_{\ast\ast}) \cap J^+(p) \setminus \{ q_{\ast\ast} \}$}}
\resizebox{3.1in}{!}{\includegraphics{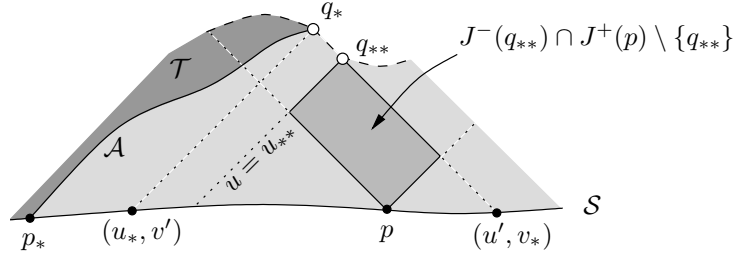}}
}
\caption{\emph{If $q_\ast$ does not have infinite $v$-coordinate, then we can find a rectangle $J^-(q_{\ast\ast}) \cap J^+(p) \setminus \{ q_{\ast\ast} \} \subset \mathcal{R} \cup \mathcal{A}$ as shown and derive a contradiction to the extension principle of \cite{D2}.}}
\label{thmpf}
\end{center}
\end{figure}

We have now shown the existence of a smooth, connected marginally trapped tube $\mathcal{A} = \del^+\mathcal{V}$, which is everywhere spacelike or outgoing-null,  intersects $\mathcal{S}$ at $p_\ast$, and exists for all $v \in [ v(p_\ast), \infty)$.  
The final piece of the proof is to show that $\mathcal{A}$ is in fact asymptotic to the event horizon $\mathcal{H}$.
Now, $\mathcal{H}$ is an outgoing null ray $\{ U \} \times [ V, \infty)$, where $\{ (U,V) \} = \mathcal{H} \cap \mathcal{S}$.
Recall that $r_+ := \sup_{p \in \mathcal{H}}r(p)$; in fact, since $\del_v r \geq 0$ on $\mathcal{H}$, we have $r_+ = \lim_{v\rightarrow\infty} r(U,v)$.
Furthermore, Lemma 3 of \cite{D2} shows that $r \leq 2 M_f$ on $\mathcal{H}$, where $M_f$ is the final Bondi mass. 
Then since $M_f \leq M$ and $\lambda > 2$, we have $r_+ \leq 2 M_f \leq 2M < \lambda M = R$.
In particular, this implies that the timelike curve $\{ r = R \}$ lies in the exterior of the black hole $\mathcal{B}$ and tends to $i^+$, never intersecting $\mathcal{H}$.
Since $\mathcal{H} \subset \mathcal{R} \cup \mathcal{A}$ and $\mathcal{V} = \mathcal{R} \cap \{ r \leq R \}$, it follows that $\mathcal{H} \subset \overline{\mathcal{V}}$.
Also, since $\del_u r < 0$ in $\mathcal{Q}$, we have $r \leq r_+$ everywhere in $J^+(\mathcal{H}) \cap \mathcal{Q}$.

We showed above that the outer ``endpoint" of the marginally trapped tube $\mathcal{A}$ is $q_\ast = (u_\ast, \infty)$.
One way of interpreting the idea that $\mathcal{A}$ should be asymptotic to $\mathcal{H}$ in this regime is to require that $u_\ast = U$; then $\mathcal{A}$ and $\mathcal{H}$ are asymptotic as curves in the underlying $\mathbb{M}^{1+1}$.
Alternately, perhaps more geometrically, one might instead require that $\lim_{p \rightarrow q_\ast, \, p \, \in \mathcal{A} } \left( r(p) \right) = r_+$. 
(One readily computes that $r$ is non-decreasing along $\mathcal{A}$ toward $q_\ast$, and since $r \leq r_+$ on $\mathcal{A}$, this limit must exist.)
Note that both of these interpretations of ``asymptotic" allow the possibility that $\mathcal{A}$ coincides with $\mathcal{H}$ for all $\widetilde{v} \geq v^\prime$, some $v^\prime \geq V$.
We shall show that $\mathcal{A}$ is indeed asymptotic to $\mathcal{H}$ in both of these senses; both follow directly from the bootstrap estimates in $\mathcal{V}$ and equation \eqref{druv} below.

For the first statement, we argue by contradiction: suppose $u_\ast > U$, and fix a value $U < \overline{u} < u_\ast$.
By construction, the infinite rectangle $[ U, \overline{u} ] \times [V, \infty) \subset \mathcal{V}$.
Now, differentiating \eqref{mdef} with respect to $u$, combining the result with equations  \eqref{eq1} and \eqref{eq3}, and finally solving for $\del^{2}_{uv} r$, we obtain
\begin{equation}
\del^{2}_{uv} r  =  -\textstyle\frac{1}{2}\Omega^2 r^{-2} \left( m - r^3V(\phi) \right) =  2  \kappa r^{-2} (\del_u r)\alpha.  \label{druv}
\end{equation}
Since $r(u,v) \leq r_+$ in $[ U, \overline{u} ] \times [V, \infty) \subset \mathcal{V}$,  rearranging \eqref{druv} and using this bound together with the fact that inequalities \eqref{kappa} and \eqref{alpha} hold in $\mathcal{V}$, we have
\begin{eqnarray*} \del_v \log (-\del_u r) & = & 2 \kappa r^{-2} \alpha \\
& > & 2 \kappa_1 r_+^{-2} \alpha_1
\end{eqnarray*}
everywhere in  $[ U, \overline{u} ] \times [V, \infty)$.
For any $U \leq u \leq \overline{u}$, integrating along an outgoing null segment $\{ u \} \times [V, v]$ yields
\begin{equation*}  - \del_u r(u,v) > - \del_u r(u,V) \,e^{2 \kappa_1 r_+^{-2} \alpha_1 (v-V)}. \end{equation*}
Since $[ U, \overline{u}] \times \{ V \}$ is compact, we may assume $- \del_u r(u,V) \geq b$ for all $U \leq u \leq \overline{u}$, some $b > 0$, so for all $(u,v) \in [U, \overline{u}] \times [V, \infty)$, we have
\[ - \del_u r(u,v) > b e^{2 \kappa_1 r_+^{-2} \alpha_1 (v-V)}. \]
Now integrating this along an ingoing null ray segment $[U, \overline{u}] \times \{ v \}$ for any $v \geq V$, we arrive at
\[ r(\overline{u},v) < r(U,v) - b e^{2 \kappa_1 r_+^{-2} \alpha_1 (v-V)} (\overline{u}-U). \]
But since $\overline{u} - U > 0$ and $r(U,v) \rightarrow_{v\rightarrow\infty} r_+ < \infty$, the right-hand side tends to $-\infty$ as $v \rightarrow \infty$, while the left-hand side remains positive, a contradiction.
So in fact $u_\ast = U$, and  $\mathcal{A}$ and $\mathcal{H}$ are asymptotic as curves in $\mathbb{M}^{1+1}$.

To show that $\mathcal{A}$ and $\mathcal{H}$ are asymptotic in the second sense described above, we again argue by contradiction:  
parametrize $\mathcal{A}$ by $\{ (u(v), v) \}$, $v(p_\ast) \leq v < \infty$ --- this is possible since $\mathcal{A}$ is everywhere spacelike or outgoing-null --- and suppose that $r(u(v), v) \rightarrow r_+ - \delta$ as $v \rightarrow \infty$, some constant $\delta > 0$.
Then we again apply bounds \eqref{kappa} and \eqref{alpha},  integrate \eqref{druv} along the ingoing null segment $[U, u(v)] \times \{ v \}$,   and use the fact that $\del_v r (u(v), v) \equiv 0$ to obtain the inequality
\[ \del_v r (U, v) > 2  \kappa_1 r_+^{-2} \alpha_1 \left( - r(u(v), v) + r(U,v) \right) \]
for all $v \geq V$.
Taking the $\liminf$ of both sides as $v \rightarrow \infty$, we have
\[ \liminf_{v\rightarrow \infty} \left( \del_v r (U, v) \right) > 2 \kappa_1 r_+^{-2} \alpha_1 \left( - (r_+ - \delta) + r_+ \right) =  2 \kappa_1 r_+^{-2} \alpha_1 \delta > 0, \]
a contradiction to the fact that $\lim_{v\rightarrow \infty} \left( r(U,v) \right) =  r_+ < \infty$.
Thus $\mathcal{A}$ and $\mathcal{H}$ do indeed tend to the same limiting radius as $v\rightarrow\infty$ as well.

\medskip

%%%%%%%%%%%%%%%%%%%%%%%%%%%%%%%%%%%%%%%%%%%%%%%%%%%%%%%%%%%%%%%%%%%%%%%%%%%%%%%%%%%%%%%%%%%%%%%%%%%%%%%%%%%%%%%%%
%
%
%					RETURNING TO THE BOOTSTRAP
%
%%%%%%%%%%%%%%%%%%%%%%%%%%%%%%%%%%%%%%%%%%%%%%%%%%%%%%%%%%%%%%%%%%%%%%%%%%%%%%%%%%%%%%%%%%%%%%%%%%%%%%%%%%%%%%%%%%

We have now completed the proof of the theorem except for the bootstrap itself.  
That is, it remains to show, one, that strict inequalities \eqref{redshift}-\eqref{drdv} hold along $\mathcal{S}_o  = \mathcal{S} \cap \{ r_1 \leq r \leq R \}$, and two, that if non-strict versions of \eqref{redshift}-\eqref{drdv} hold for all points  within $\{ r \leq R \}$ to the causal past of a point $(u,v) \in \mathcal{R} \cap \{ r \leq R \}$, then in fact strict versions of the six inequalities must hold at $p$. 

\medskip
First, from \eqref{Cdef}, \eqref{Lambdadef}, \eqref{kdef}, \eqref{rhodef}, and  \eqref{C1}, we have $C_1 > C$, so it follows from \eqref{lambdadef2} that inequality \eqref{redshift} is satisfied on $\mathcal{S}_o$.
To see that \eqref{phibound} holds on $\mathcal{S}_o$, we apply \eqref{phicond}, \eqref{r1inf}, \eqref{mcond}, \eqref{rhodef}, and \eqref{Chat}:
\begin{equation} |\phi| \leq \frac{\epsconst}{\sqrt{r}} \leq \frac{\epsconst}{\sqrt{r_1}} = \frac{\epsconst}{\sqrt{2m_1}}
\leq \frac{\epsconst}{\sqrt{2\rhoconst M}} < \frac{\epsconst}{\sqrt{2 \sigma^2 M}} < \widehat{C}. \label{phiSo}
\end{equation}
That \eqref{kappa} is satisfied along all of $\mathcal{S}$ follows immediately from \eqref{kappa0def} and  \eqref{kappa1}.
And by construction, $\mathcal{S}_o \subset \mathcal{R}$, so clearly \eqref{drdv} holds on $\mathcal{S}_o$ as well.

In order to see that $\mathcal{S}_o \subset \mathcal{V}$, it remains to show that the bounds \eqref{potential} and \eqref{alpha} are satisfied.
In fact, we show something much stronger: that strict inequalities \eqref{potential} and \eqref{alpha} hold wherever a non-strict version of \eqref{phibound} does. 

First observe that since $V^\prime(0) = 0$, the mean value theorem says that for any $x \neq 0$,  there exists $\xi$ with $0 < |\xi| < |x|$ such that
\[ \left| \frac{V^\prime(x)}{x} \right| = \left| V^{\prime\prime}(\xi) \right| \leq B, \]
and similarly, since $V(0) = 0$, there exists $\xi^\prime$ with $0 < |\xi^\prime| < |x|$ such that
\[ \frac{V(x)}{|x|} = |V^{\prime}(\xi^\prime)| \leq B |\xi^\prime| < B|x|,  \]
so that
\[ V(x) \leq B x^2 \]
for all $x$. 

Next, note that	\eqref{Chat} and \eqref{kdef2} together imply that $\widehat{C} < \left( 16 M^2 B k^3 \right)^{-\frac{1}{2}}$.
Thus at any point in $\overline{\mathcal{R}} \cap \{r \leq R \}$ where a non-strict version of \eqref{phibound} holds, we compute that
\begin{eqnarray*} \alpha & = & m - r^3 V(\phi) \\
& \geq & m_1 - R^3 B \phi^2 \\
& \geq & \rhoconst M - (\lambda M)^3 B \widehat{C}^2 \\
%& \geq & \rhoconst M - \lambda^3 M^3 B \left( \frac{\epsconst}{\sigma \sqrt{2M}} + \sqrt{\textstyle\frac{1}{2}\log 2 \log k} \right)^2 \\
& >    & \rhoconst M - \lambda^3 M^3 B \left(\frac{1}{16M^2 B k^3} \right)\\
%& =    & \rhoconst M - \lambda^3 M \frac{1}{16 k^3} \\
& =    & \rhoconst M - \frac{1}{2}\left(\displaystyle\frac{\lambda}{2k} \right)^3 M  \\
& > & \rhoconst M - \textstyle\frac{1}{2} \rhoconst^3 M  \\
& > & \rhoconst M - \textstyle\frac{1}{2}\rhoconst M  \\
& = & \alpha_1,
\end{eqnarray*}
where in addition to \eqref{phibound} we have used \eqref{rhodef} and the fact that $\rhoconst > \rhoconst^3$ (since $\rhoconst < 1$).
Also, 
\begin{eqnarray*} | V^\prime(\phi) | & \leq & B | \phi | \\
& \leq & B \widehat{C} \\
%& \leq & B \left( \frac{\epsconst}{\sigma \sqrt{2M}} + \sqrt{\textstyle\frac{1}{2}\log 2 \log k} \right)  \\
& < &  B \left(\frac{1}{4M \sqrt{B k^3}} \right) \\
& < & \frac{\sqrt{B}}{4M} \\
& = & \overline{C}.
\end{eqnarray*}
where in the second to last line we used the fact that $k > 1$.
Thus \eqref{potential} and \eqref{alpha} hold on $\mathcal{S}_o$, since we have already shown that \eqref{phibound} does.  
Since the six inequalities \eqref{redshift}-\eqref{drdv} are all satisfied along $\mathcal{S}_o$, we have $\mathcal{S}_o \subset \mathcal{V}$, and $\mathcal{V} \neq \emptyset$.

\medskip

%%%%%%%%%%%%%%%%%%%%%%%%%%%%%%%%%%%%%%%%%%%%%%%%%%%%%%%%%%%%%%%%%%%%%%%%%%%%%%%%%%%%%%%%%%%%%%%%%%%%%%%%%%%%%%%%%
%
%
%					RETRIEVING KAPPA
%
%%%%%%%%%%%%%%%%%%%%%%%%%%%%%%%%%%%%%%%%%%%%%%%%%%%%%%%%%%%%%%%%%%%%%%%%%%%%%%%%%%%%%%%%%%%%%%%%%%%%%%%%%%%%%%%%%%

It remains to retrieve \eqref{redshift}-\eqref{drdv} in $\overline{\mathcal{V}} \cap \mathcal{R}$.
Clearly \eqref{drdv} is automatically satisfied by definition of $\mathcal{R}$, so the five nontrivial inequalities to consider are \eqref{redshift}-\eqref{kappa}.

We begin with \eqref{kappa}.
To retrieve the desired inequality in $\overline{\mathcal{V}} \cap \mathcal{R}$, we estimate $\kappa$ in two pieces, first using an energy estimate to get a stronger bound on $\mathcal{S}_o \cup \{ r = R \}$, then integrating in $u$ and using the bootstrap assumptions to obtain the desired one in $\overline{\mathcal{V}}$.

For the bound on $\{ r = R \}$, first recall that by hypothesis, we have $\kappa \geq \kappa_0$ on all of $\mathcal{S}$. 
Also observe that in the exterior region $\{ r \geq R = \lambda M \}$, we have
\begin{equation} \label{fracbound} 1-\frac{2m}{r} \geq 1-\frac{2M}{R} = 1-\frac{2}{\lambda} > 0.  \end{equation}
Now, rearranging equation \eqref{eq1}, we obtain
\begin{equation} \del_u \kappa = \left( \frac{r\del_u \phi}{\del_u r}\right)^2 \left( \frac{\del_u r}{r} \right) \kappa, \label{dukappa}
\end{equation}
and integrating along the ingoing null ray $[ u^\prime, u] \times \{ v \}$, where the point $(u,v) \in \{ r = R = \lambda M \}$ and $(u^\prime, v ) \in \mathcal{S}$, we apply \eqref{fracbound}, the energy estimate \eqref{u-energy}, as well as \eqref{mcond}, \eqref{duneg}, and \eqref{elldef}, to get
\begin{align*} \kappa(u,v)  \,\,\, = & \quad
\kappa(u^\prime, v) \exp{\left[\int_{u^\prime}^u  \left( 1-\frac{2m}{r}\right) \left( \frac{r\del_u \phi}{\del_u r}\right)^2  \del_u r  \cdot
\frac{1}{\left( 1-\frac{2m}{r} \right) r} (\tilde{u}, v) d\tilde{u} \right]} \\
 \geq & \quad
\kappa_0 \exp{\left[\frac{1}{\left( \lambda-2 \right) M}  \int_{u^\prime}^u  \left( 1-\frac{2m}{r}\right) \left( \frac{r\del_u \phi}{\del_u r}\right)^2  \del_u r  
(\tilde{u}, v) d\tilde{u} \right]} \\
 \geq & \quad
\kappa_0 \exp{\left[-\frac{2\left(1 - \frac{m_1}{M}\right)}{\left( \lambda-2 \right)}  \right]} \\
\geq & \quad
\kappa_0 \cdot e^{- 2\ell/\lambda  }.
\end{align*}
Since $\kappa \geq \kappa_0$ on $\mathcal{S}$ and $e^{- 2\ell/\lambda  } < 1$, we thus have $\kappa \geq \kappa_0 \cdot e^{- 2\ell/\lambda}$ everywhere on $\mathcal{S}_o \cup \{ r = R \}$.

Now, given any point $(u,v) \in \overline{\mathcal{V}} \cap \mathcal{R}$, the ingoing null ray to the past of $(u,v)$ must intersect $\overline{\mathcal{S}}_o \cup \{ r = R \}$, say at the point $(u^{\prime\prime}, v)$.
Again we integrate \eqref{dukappa},  this time applying the bootstrap inequality \eqref{redshift} along with \eqref{duneg} and \eqref{C1} to obtain 
\begin{align*} \kappa(u,v) \quad = & \quad \kappa(u^{\prime\prime}, v) \exp{\left[ \int_{u^{\prime\prime}}^u \left( \frac{r \del_u \phi}{\del_u r} \right)^2 \left( \frac{\del_u r}{r} \right)(\bar{u}, v)  d\bar{u} \right]} \nonumber\\
 \geq & \quad\kappa_0 \cdot e^{- 2\ell/\lambda  } \exp \left[ C_1^2 \,\log\left( \frac{r(u,v)}{r(u^{\prime\prime}, v)}\right) \right]  \nonumber \\
 \geq & \quad\kappa_0 \cdot e^{- 2\ell/\lambda  } \exp \left[ C_1^2 \,\log\left( \frac{2\rhoconst}{\lambda}\right) \right]  \nonumber \\
 = & \quad\kappa_0 \cdot e^{- 2\ell/\lambda  } \exp \left[ \left( \frac{\log 2}{2 \log (\frac{\lambda}{2\rhoconst}) } \right) \,\log\left( \frac{2\rhoconst}{\lambda}\right) \right]  \nonumber \\
 = & \quad\kappa_0 \cdot e^{- 2\ell/\lambda  } \cdot \frac{1}{\sqrt{2}}  \nonumber \\
 > & \quad \kappa_1.
\end{align*}

Thus we have retrieved and in fact improved \eqref{kappa} in $\overline{\mathcal{V}} \cap \mathcal{R}$.
Separately, note also that it follows from \eqref{dukappa} that $\del_u \kappa < 0$ everywhere in $\mathcal{Q}$ (since $\kappa > 0$  by \eqref{duneg}), and thus our assumption that $\kappa \leq 1$ on $\mathcal{S}$ implies that $\kappa \leq 1$ in all of $\mathcal{Q}$.

\medskip

%%%%%%%%%%%%%%%%%%%%%%%%%%%%%%%%%%%%%%%%%%%%%%%%%%%%%%%%%%%%%%%%%%%%%%%%%%%%%%%%%%%%%%%%%%%%%%%%%%%%%%%%%%%%%%%%%
%
%
%					RETRIEVING PHI, & HENCE ALPHA AND POTENTIAL
%
%%%%%%%%%%%%%%%%%%%%%%%%%%%%%%%%%%%%%%%%%%%%%%%%%%%%%%%%%%%%%%%%%%%%%%%%%%%%%%%%%%%%%%%%%%%%%%%%%%%%%%%%%%%%%%%%%%

Next, as was shown previously, in order to retrieve \eqref{phibound}, \eqref{potential}, and \eqref{alpha} in $\overline{\mathcal{V}} \cap \mathcal{R}$, it suffices to retrieve \eqref{phibound}.
As with \eqref{kappa}, retrieving \eqref{phibound} requires two steps:  we first estimate $|\phi|$ on $\mathcal{S}_o \cup \{ r = R \}$, again utilizing an energy estimate, then use that bound to obtain $|\phi| < \widehat{C}$ in $\overline{\mathcal{V}} \cap \mathcal{R}$.

From \eqref{phiSo}, we know already that $|\phi| <  \epsconst/\sqrt{2 \sigma^2 M}$ on $\mathcal{S}_o$.  
To see that the same bound holds along $\{ r = R \}$, we integrate $|\del_u \phi|$ along an ingoing null ray emanating from $\mathcal{S}$.
Suppose $(u,v^\prime) \in \{ r = R \}$ and $\{ (u^\prime, v^\prime) \} = \mathcal{S} \cap \{ v = v^\prime \}$.  
Then since the ingoing null ray segment $[u^\prime, u] \times \{ v^\prime \} \subset \{ r \geq R = \lambda M \} \subset \mathcal{R}$, integrating $\del_u \phi$ along it and applying \eqref{phicond}, Cauchy-Schwarz, \eqref{fracbound}, the energy estimate \eqref{u-energy}, as well as \eqref{elldef}, \eqref{Ldef}, \eqref{sigmadef} and \eqref{lambdadef},  yields
\begin{align*}  |\phi(u,v)|  \,\,\, \leq & \quad |\phi(u^\prime,v)| + \int_{u^\prime}^u |\del_u \phi (\tilde{u}, v) | d\tilde{u}   \\
 \leq & \quad \epsconst \cdot r(u^\prime,v)^{-\frac{1}{2}} \\
& \qquad + \sqrt{\int_{u^\prime}^u - \left(1 - \frac{2m}{r}\right)\left( \frac{r \del_u \phi}{\del_u r} \right)^2 \del_u r (\tilde{u}, v)  d\tilde{u}}  \\
& \qquad\quad \cdot \,\sqrt{ \int_{u^\prime}^u - \frac{\del_u r}{\left(1-\frac{2m}{r}\right)r^2} (\tilde{u}, v)  d\tilde{u}  } \\
 \leq & \quad \epsconst \cdot r(u^\prime,v)^{-\frac{1}{2}} 
+ \sqrt{2(M-m_1)} \sqrt{ \frac{1}{1-\frac{2}{\lambda}}\int_{u^\prime}^u - \frac{\del_u r}{r^2} (\tilde{u}, v)  d\tilde{u}  } \\
 < & \quad \epsconst \cdot r(u^\prime,v)^{-\frac{1}{2}} 
+ \sqrt{2(M-m_1)} \sqrt{ \frac{1}{1-\frac{2}{\lambda}} \cdot \frac{1}{R}   } \\
 \leq  & \quad \frac{\epsconst}{\sqrt{\lambda M}} + \sqrt{ \frac{2(1-\frac{m_1}{M})}{\lambda-2}  } \\
\leq & \quad \frac{\epsconst}{\sqrt{\lambda M}} + \sqrt{ \frac{2\ell}{\lambda}  } \\
< & \quad \frac{1}{\sqrt{\lambda M}}\left( \epsconst + \sqrt{ 2 L M   }  \right) \\
< & \quad \frac{\epsconst }{\sigma \sqrt{2 M}}.
\end{align*}
Thus we have $|\phi| < \displaystyle\frac{\epsconst }{\sigma \sqrt{2 M}}$ everywhere along $\mathcal{S}_o \cup \{ r = R \}$.

Now consider a point $(u,v) \in \overline{\mathcal{V}} \cap \mathcal{R}$ with past ingoing null ray intersecting $\overline{\mathcal{S}}_o \cup \{ r = R \}$ at the point $(u^{\prime\prime}, v)$. 
Then:
\begin{align*} \left| \phi(u,v) \right|  \quad \leq & \quad \left| \phi(u^{\prime\prime},v) \right| + \int_{u^{\prime\prime}}^u |\del_u \phi  (\bar{u}, v) | d\bar{u} \nonumber \\ 
 < & \quad \frac{\epsconst }{\sigma \sqrt{2 M}} + \int_{u^{\prime\prime}}^u \left|\frac{r\del_u \phi}{\del_u r}\right| \left|\frac{\del_u r}{r}\right| (\bar{u}, v)  d\bar{u} \nonumber \\
 \leq & \quad \frac{\epsconst }{\sigma \sqrt{2 M}} + C_1 \log \left(\frac{r(u^{\prime\prime}, v)}{r(u,v)}\right)   \nonumber \\
 \leq & \quad \frac{\epsconst }{\sigma \sqrt{2 M}} + C_1 \log \left(\frac{R}{r_1}\right) \nonumber \\
%    = & \quad \frac{\epsconst }{\sigma \sqrt{2 M}} + C_1 \log \left(\frac{\lambda M }{2 \rho M}\right) \nonumber \\
    = & \quad \frac{\epsconst }{\sigma \sqrt{2 M}} + \sqrt{ \frac{\log 2}{2 \log \left( \frac{\lambda}{2 \rhoconst} \right) } } \cdot \log \left(\frac{\lambda }{2 \rho}\right) , \nonumber  \\
    = & \quad \frac{\epsconst }{\sigma \sqrt{2 M}} + \sqrt{ \textstyle\frac{1}{2}\log{2} \log{\left( \frac{\lambda}{2 \rhoconst} \right)}  }  , \nonumber \\
    < & \quad \frac{\epsconst }{\sigma \sqrt{2 M}} + \sqrt{ \textstyle\frac{1}{2}\log{2} \log{k}  }   \nonumber \\
    = & \quad \widehat{C},
\end{align*}
where we have used \eqref{redshift},  \eqref{C1}, \eqref{rboundseq} and \eqref{rhodef}, in addition to our bound for $|\phi|$ on $\mathcal{S}_o \cup \{ r = R \}$.  
Thus \eqref{phibound} holds everywhere in $\overline{\mathcal{V}} \cap \mathcal{R}$, and hence so do \eqref{potential} and \eqref{alpha}.

\medskip

%%%%%%%%%%%%%%%%%%%%%%%%%%%%%%%%%%%%%%%%%%%%%%%%%%%%%%%%%%%%%%%%%%%%%%%%%%%%%%%%%%%%%%%%%%%%%%%%%%%%%%%%%%%%%%%%%
%
%
%					RETRIEVING REDSHIFT
%
%%%%%%%%%%%%%%%%%%%%%%%%%%%%%%%%%%%%%%%%%%%%%%%%%%%%%%%%%%%%%%%%%%%%%%%%%%%%%%%%%%%%%%%%%%%%%%%%%%%%%%%%%%%%%%%%%%

Finally, it remains only to retrieve \eqref{redshift}.  
Combining equations \eqref{sfequv} and \eqref{druv} and rearranging terms, we derive
\begin{equation} \label{redshiftevol} \del_v \left( \frac{r \del_u \phi}{\del_u r} \right)
 = -\frac{2\kappa\alpha}{r^2} \left( \frac{r \del_u \phi}{\del_u r} \right) + r\kappa V^\prime(\phi) - \del_v \phi.
\end{equation} 
Now, the past boundary of $\mathcal{V}$ lies in $\left( \mathcal{S} \cap \{ r_1 \leq r \leq R \} \right) \cup \{ r = R\}$, and since $\{ r = R\}$ is everywhere timelike, the outgoing null segment to the past of a point $(u, v) \in \overline{\mathcal{V}}$  must intersect $\overline{\mathcal{S}_o}$, say at the point $(u, v^\prime)$. 
Then integrating \eqref{redshiftevol} along the null segment $\{ u \} \times [ v^\prime, v ] \subset \overline{\mathcal{V}}$, we have
\begin{align}
\left| \frac{r \del_u \phi}{\del_u r} (u,v)\right|  \quad \leq & \quad e^{\int_{v^\prime}^v -\frac{2\kappa \alpha}{r^2}(u,\tilde{v})d\tilde{v}} \left| \frac{r \del_u \phi}{\del_u r} (u, v^\prime) \right| \nonumber\\
& \quad + \int_{v^\prime}^v e^{\int_{\bar{v}}^v -\frac{2\kappa \alpha}{r^2}(u,\tilde{v})d\tilde{v} } \left| r\kappa V^\prime(\phi) - \del_v\phi \right|(u, \bar{v}) d\bar{v} \nonumber \\
 \leq & \quad C 
+ \int_{v^\prime}^v e^{-\frac{2\kappa_1 \alpha_1}{R^2}(v - \bar{v})} \left( R\overline{C} +  |\del_v\phi| \right)(u, \bar{v}) d\bar{v} \nonumber \\
 \leq & \quad C + \frac{\overline{C} R^3}{2\kappa_1 \alpha_1} \left( 1 -  e^{-\frac{2\kappa_1 \alpha_1}{R^2}(v - v^\prime)} \right)  \nonumber \\
& \quad + \sqrt{\int_{v^\prime}^v  e^{-\frac{4\kappa_1 \alpha_1}{R^2}(v - \bar{v})} r^{-2}\kappa  (u, \bar{v}) d\bar{v}} 
\sqrt{\int_{v^\prime}^v {(\del_v\phi)^2 r^2 \kappa^{-1}  (u, \bar{v}) d\bar{v} }} \nonumber \\
 \leq & \quad C +   \frac{\overline{C}R^3}{2\kappa_1 \alpha_1} \nonumber \\
& \quad + \sqrt{r_1^{-2} \cdot  \frac{R^2}{4\kappa_1 \alpha_1} \left( 1 - e^{-\frac{4\kappa_1 \alpha_1}{R^2}(v - v^\prime)} \right) }
\cdot \sqrt{2(M - m_1)} \nonumber \\
 \leq & \quad C + \frac{\overline{C} R^3}{2\kappa_1 \alpha_1} + \sqrt{ \frac{R^2(M - m_1)}{2r_1^2\kappa_1 \alpha_1}}, \label{est1}
\end{align}
where we have used \eqref{lambdadef2}, the bootstrap inequalities \eqref{kappa}, \eqref{alpha}, and \eqref{potential},  Cauchy-Schwarz, \eqref{rboundseq}, the energy estimate \eqref{v-energy}, and the fact that $\kappa \leq 1$.
Now, the second term here is
\begin{eqnarray*} \frac{\overline{C} R^3}{2\kappa_1 \alpha_1} 
& = & \frac{\sqrt{B} }{ 4M } \cdot \frac{(\lambda M)^3}{ 2(\frac{1}{2} \kappa_0 e^{-2\ell/\lambda} )(\frac{1}{2} \rhoconst M) } \\
& = &  \frac{\sqrt{B} M}{\kappa_0 e^{-2\ell/\lambda} }\left( \frac{\lambda^3}{2 \rhoconst} \right) \\
& < &  \frac{\sqrt{B} M e^{2L/\lambda} }{\kappa_0 }\left( k \lambda^2  \right) \\
& < &  \frac{9\sqrt{B} M e^L }{\kappa_0 }, 
\end{eqnarray*}
where we have applied \eqref{Coverline}, \eqref{Rdef}, \eqref{kappa1}, \eqref{alpha1}, \eqref{rhodef}, and the facts that $\ell < L$, $\lambda > 2$, $k < 1$, and $\lambda < 3$;
and the third one is
\begin{eqnarray*}  \sqrt{ \frac{R^2(M - m_1)}{2r_1^2\kappa_1 \alpha_1}} 
& \leq & \sqrt{ \frac{(\lambda M)^2(M - \rhoconst M)}{2(2\rhoconst M)^2 (\frac{1}{2} \kappa_0 e^{-2\ell/\lambda} ) (\frac{1}{2} \rhoconst M) }} \\
& = & \frac{\lambda }{2\rhoconst } \cdot \sqrt{ \frac{2(1 - \rhoconst) e^{2\ell/\lambda} }{\rhoconst \kappa_0  }}   \\
& < & \frac{\lambda }{2\rhoconst } e^{L/2}   \sqrt{ \frac{2(1 - \rhoconst) }{\rhoconst \kappa_0  }}   \\
& < & k e^{L/2}   \sqrt{ \frac{2L }{ \sigma^2 \kappa_0  }}   \\
& < & \frac{ 3 e^{L/2} }{ \sigma }  \sqrt{ \frac{ L }{ 2 \kappa_0  }},
\end{eqnarray*}
where this time we used \eqref{Rdef}, \eqref{kappa1}, \eqref{alpha1}, $\ell < L$, $\lambda > 2$, various consequences of \eqref{rhodef}, and $k < \frac{3}{2}$.
Thus, resuming from \eqref{est1}, we have
\begin{eqnarray*}
\left| \frac{r \del_u \phi}{\del_u r} (u,v)\right| 
& < & C + \frac{9 \sqrt{B} M e^L }{\kappa_0 } + \frac{ 3 e^{L/2} }{ \sigma }  \sqrt{ \frac{ L }{ 2 \kappa_0  }} \\
& = & \Lambda \\
& < & \sqrt{\frac{\log 2 }{2 \log k }} \\
& < & \sqrt{\frac{\log 2 }{2 \log \left( \frac{\lambda}{2\rho} \right) }} \\
& = & C_1,
 \end{eqnarray*}
having applied \eqref{Lambdadef}, \eqref{kdef}, and \eqref{rhodef}.
So \eqref{redshift} is retrieved in $\overline{\mathcal{V}} \cap \mathcal{R}$.

\medskip

%%%%%%%%%%%%%%%%%%%%%%%%%%%%%%%%%%%%%%%%%%%%%%%%%%%%%%%%%%%%%%%%%%%%%%%%%%%%%%%%%%%%%%%%%%%%%%%%%%%%%%%%%%%%%%%%%
%
%
%					CONCLUSION
%
%%%%%%%%%%%%%%%%%%%%%%%%%%%%%%%%%%%%%%%%%%%%%%%%%%%%%%%%%%%%%%%%%%%%%%%%%%%%%%%%%%%%%%%%%%%%%%%%%%%%%%%%%%%%%%%%%%

We have now retrieved all of \eqref{redshift}-\eqref{drdv} in $\overline{\mathcal{V}} \cap \mathcal{R}$, completing the bootstrap and hence the proof.

\end{proof}

%%%%%%%%%%%%%%%%%%%%%%%%%%%%%%%%%%%%%%%%%%%%%%%%%%%%%%%%%%%%%%%%%%%%%%%%%%%%%%%%%%%%%%%%%%%%%%%%%%%%%%%%%%%%%%%%%
%
%
%					ACKNOWLEDGMENTS
%
%%%%%%%%%%%%%%%%%%%%%%%%%%%%%%%%%%%%%%%%%%%%%%%%%%%%%%%%%%%%%%%%%%%%%%%%%%%%%%%%%%%%%%%%%%%%%%%%%%%%%%%%%%%%%%%%%%

\section*{Acknowledgments}

The author wishes to thank Mihalis Dafermos for several very helpful conversations and suggestions, and Jan Metzger for pointing out a small but important error in a previous version of the paper.
Also, the author is grateful to the Mittag-Leffler Institute in Djursholm, Sweden, for providing support and an excellent research environment for this work during the Fall 2008 program \emph{Geometry, Analysis, and General Relativity}.

\bibliography{mybib}{}
\bibliographystyle{amsplain}

\end{document}